\documentclass[citeauthoryear]{llncs}
\pdfoutput=1 
\usepackage{llncsdoc}
\usepackage[english]{babel}
\usepackage[utf8]{inputenc}
\usepackage[ngerman,colorlinks=true]{hyperref}
\usepackage{amsmath,amssymb,latexsym}
\usepackage{wasysym,stmaryrd}

\usepackage{tabularx}

\usepackage{float}
\usepackage{wrapfig}
\usepackage{subcaption}
\usepackage{graphicx}

\usepackage{color}
\usepackage[all]{xy}
\input{xy} \xyoption{all}

\newcommand{\PotDir}{\ensuremath{{\mathcal{P}}^{(1,2)} } }
\newcommand{\Pot}{\ensuremath{\mathcal{P} } }
\newcommand{\Potfin}{\Pot_{fin}}
\newcommand{\Potom}{{\ensuremath{\Pot^\omega}} }
\newcommand{\PPb}{\ensuremath{\mathfrak{P}} }
\newcommand{\PPa}{\ensuremath{\mathbb{P}} }

\newcommand{\rGF}{\ensuremath{\mathbb{F}} } 

\newcommand{\Nat}{\ensuremath{{\mathbb{N}}} }
\newcommand{\Int}{\ensuremath{{\mathbb{Z}}} }

\newcommand{\impl}{\Rightarrow}

\newcommand{\nach}[1]{\stackrel{#1}{\to}}
\newcommand{\von}[1]{\stackrel{#1}{\gets}}

\newcommand{\injct}[1]{\stackrel{#1}{\hookrightarrow} }

\newcommand{\categ}[1]{\ensuremath{\mathbf{ #1} } }

\newcommand{\cC}{\categ{C}}
\newcommand{\cD}{\categ{D}}

\newcommand{\cSets}{\categ{Sets}}
\newcommand{\cSetsF}[1]{\categ{{Sets}_{#1}}}
\newcommand{\coalg}[1]{{\categ{Coalg}_{#1}}}

\newcommand{\PPaGraph}{\PPa\_\categ{Graph}}
\newcommand{\PPbGraph}{\PPb\_\categ{Graph}}
\newcommand{\PotomGraph}{\Potom\_\categ{Graph}}

\newcommand{\functor}[1]{\textsf{#1}}

\newcommand{\funF}{\functor{F}}
\newcommand{\funG}{\functor{G}}
\newcommand{\funH}{\functor{H}}

\newcommand{\funFinal}{\mathbf{1}  }
\newcommand{\funX}{\functor{X}}

\newcommand{\finObj}{\ensuremath{\mathbf{1}  } }
\newcommand{\finMor}{\ensuremath{\mathbf{!}  } }

\newcommand{\M}{\ensuremath{\mathcal{M}} }

{\bfseries}{\itshape}

\newcommand{\node}{\ensuremath{con} }
\newcommand{\st}{\ensuremath{ngb}}
\newcommand{\stp}[1]{\ensuremath{{ngb_{#1}}^+} }
\newcommand{\NN}{\ensuremath{N} }
\newcommand{\EN}{\ensuremath{E} }

\newcommand{\aV}{\ensuremath{V} }
\newcommand{\aE}{\ensuremath{E} }

\newcommand{\uned}[1]{\xymatrix{\ar@{{}{-}{*}}[r]^{#1} & }}
\newcommand{\ppf}[2]{\ensuremath{  {#1_{#2}}^{\star}  }}  

\begin{document}
\title{Towards {$\M$}-Adhesive Categories based on  Coalgebras and Comma Categories}

\author{Julia Padberg}
              
\institute{Hamburg University of Applied Sciences\\Germany\\\email{julia.padberg@haw-hamburg.de}}

\maketitle
\begin{abstract}
In this contribution we investigate several extensions of the powerset that comprise arbitrarily nested subsets, and call them superpower set. This allows the definition of graphs with possibly infinitely nested nodes. Additionally we define edges that are incident to edges. Since we use coalgebraic constructions we refer to these graphs as coalgebraic graphs. The superpower set functors are examined and then used  for the definition of $\M$-adhesive categories which are the basic categories for $\M$-adhesive transformation systems.  So, we additionally show that coalgebras
$\cSetsF{F}$ are $\M$-adhesive categories provided the  functor $F:\cSets \to \cSets$  preserves pullbacks along monomorphisms.
\end{abstract}

\tableofcontents

\section{Motivation}
\label{s.motiv}
% motiv.tex

The main motivation of this paper is the question how to define recursion on a graph's structure so that we still obtain an $\M$-adhesive transformation systems. Since the recursion construct we use in this contribution is a coalgebraic construction we consequently use the term coalgebraic graphs. \\
A shorter version can be found in \cite{Pad17}.
We start with a examples of such graphs to illustrate what we aim at.
\begin{example}[Coalgebraic graphs]
The coalgebraic graph $G_1 =(\NN_1,\EN_1,\node_1,\st_1)$ is given in Fig.~\ref{fig:ex1} and consists of  a set of nodes $\NN$, a set of edges $\EN$, a contains function $\node$ and a neighbour function $\st$. $\node $ yields the set of nodes that  for each node may contain. Those  nodes that are  mapped to themselves, are considered to be atomic.\\
Let $\NN_1= \{ n_1,n_2,n_3, n_4, n_5, n_6\}$ with:\\ 
$\node_1(n_i) =   \begin{cases}
                n_i &; 1\le i\le 3\\
	\{n_1,n_2\} &; i=4 \\
	 \{n_3\}\    &;i=5 \\
	\{n_2, \{n_2, n_3\}, n_5\} &;i =6
                         \end{cases}
$\\
The atomic nodes are $\NN_1=\{n_1,n_2,n_3\}$ and\\
$\node(\NN_1)=\{n_1,n_2,n_3, \{n_1,n_2\}, \{n_3\}\  , \{n_2, \{n_1, n_2\}, n_5\} \}$.\\
We have $n_1,n_2 \in \node_1(n_4)$ and $\{n_1,n_2\}  \in \node_1(n_6)$.

\begin{figure}[H]

	 \begin{subfigure}[b]{0.5\linewidth}
%\subfigure[$\label{fig:ex1}]{
         \includegraphics[width=\linewidth]{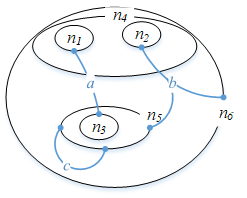}
        \caption{ Nested nodes in graph $G_1$}
				\label{fig:ex1}
    \end{subfigure} \hfill
	 \begin{subfigure}[b]{.43\linewidth}
%\subfigure[Coalgebraic edges in $G_2$ \label{fig:ex2}]{
\includegraphics[width=\linewidth]{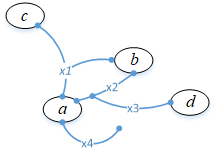}
\caption{ Nested edges in $G_2$ }
				\label{fig:ex2}
    \end{subfigure} 
		\centering
	 \begin{subfigure}[b]{.6\linewidth}
\includegraphics[width=\linewidth]{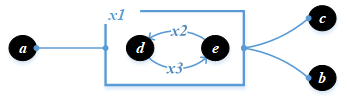}
\caption{Edges containing subgraphs  in $G_3$ }
				\label{fig:ex3}
    \end{subfigure} 
\caption{Examples of  graphs with recursive structures}
\label{fig:ex}
\end{figure}
$\EN_1=\{a,b,c\}$ with  
$ \begin{array}[t]{rr@{\; \mapsto\; }l}
	\st1:  &  a &    \{n_1,n_3\}\\
				&	b & \{n_2,n_5,n_6\} \\
				&	c&   \{n_5\}
\end{array}
$\\
Edge $b$ is a hyperedge. \\[2mm]
In Fig.~\ref{fig:ex2} the coalgebraic graph $G_2= (\NN_2,\EN_2,\st_2)$ with  atomic nodes $\NN_2=\{a,b,c,d\}$ and  edges $\EN_2=\{x1,x2,x3,x4\}$ with 
$\begin{array}[t]{rr@{\mapsto}l}
\st_2:  & x1 &  \{a,b,c\}\\
           & x2 &  \{a,b\}\\
           & x3 &  \{x2,d\}\\
           & x4 &  \{a, x4\}	
\end{array}
$\\

The edge $x_1$ is also a hyperedge.
The edge $x_3$ is attached to edge $x_2$ with $x_2 \in\st(x3)$  and  $x4$ is an unary edge that is attached to itself and is denoted by its end $n_2 \uned{x_4}$.\\[2mm]

$G_2$ can be flattened to a hypergraph with the attachment function 
$att:= \stp{2} :\EN\to \Pot(\aV)$ and
$\begin{array}[t]{rr@{\mapsto}l}
\stp{2}: & x1 &  \{a,b,c\}\\
           & x2 &  \{a,b\}\\
           & x3 &  \{a,b,d\}\\
           & x4 &  \{a\}	
\end{array}
$

$G_3$ has an egde, that contains a subgraph. In Fig.~\ref{fig:ex3} the coalgebraic graph $G_3= (\NN_3,\EN_3,\st_2)$ with  atomic nodes $\NN_2=\{a,b,c,d,e\}$ and  edges $\EN_2=\{x1,x2,x3\}$ with 
$\begin{array}[t]{rr@{\mapsto}l}
\st_2:  & x1 &  (abc,\{d,e,x2,x3\})\\
           & x2 &  (ed,\emptyset)\\
           & x3 &  (de,\emptyset)
\end{array}
$\\

\end{example}

The concepts investigated in this contribution  shall comprise 
\begin{itemize}
	\item usual (un-)directed multi-graphs,
	\item  classic hypergraphs,
	\item graphs with hyperedges as in \cite{DKH97},
	\item  bigraphs \cite{Milner06} see Sect.~\ref{ss.bigr},
	\item  and various   hierarchical graphs see Sect.~\ref{ss.hierG}.
\end{itemize}

\section{Super Power Sets}
\label{s.noderec}
% noderec.tex vormals ideas.tex

The superpower set  is achieved by recursively inserting  subsets  of the superpower set into  itself. In this contribution  we present  three possibilities:
\begin{enumerate}
	\item $\PPa$ allows atomic nodes.
	\item $\PPb$ only allows sets  of nodes.
	\item $\Potom$ layers the nesting of nodes.
\end{enumerate}

Subsequently, we investigate the properties of each construction and in Subsect.~\ref{ss.diff} we discuss the differences.

\subsection{Node Recursion based on $\PPa$}

\begin{definition}[Superpower set $\PPa$]
\label{d.ppa_set}
Given a well-founded set  $M$ and $\Pot (M) $ the power set of $M$ then we define the superpower set $\PPa(M)$
	\begin{enumerate}
	\item $M \subset \PPa(M)$ and $\Pot(M) \subset \PPa(M)$
	\item If $M' \subset \PPa(M)$ then $M' \in \PPa(M)$.
\end{enumerate}
$\PPa(M)$ is the smallest set satisfying 1. and 2.
\end{definition}
Note that this superpower set construction $\PPa$ is well-founded sets with the ordering with respect to the number of parentheses (see Appendix~\ref{s.NoetherianInd}).

\begin{lemma}[$\PPa$ is a functor]\label{l.PPa_functor}
$\PPa:\cSets \to \cSets$ is defined for finite sets as in Def.~\ref{d.ppa_set} and for functions $f:M \to N$ by
$\ppf{f}{}:\PPa(M) \to \PPa(N)$ with\\
$\ppf{f}{}(x)=
							\begin{cases}
	            f(x)  &\; ; x\in M\\
							\{\ppf{f}{} (x') \mid x' \in x\} &\; ; \text{ else}
  \end{cases}
$
\end{lemma}
We use $\ppf{f}{}$ instead of $\PPa(f) $ merely for better readability and to stress the similarities of $\PPa$, $\PPb$ and $\Potom$.

\begin{example}[Functor $\PPa$]
Given sets $M=\{u,v,w,u',v'\}$ and $N=\{ n_1,n_2,n_3, n_4, n_5, n_6\}$  with $f:M\to N$ and
$\begin{array}[t]{cr@{\mapsto}l@{\; ; \;}r@{\mapsto}l @{\; ; \;}r@{\mapsto}l}
    f: &u & n_3  & v & n_3 & w &n_1 \\
		   &u' & n_5 & v'& n_5 \\
\end{array}
$ \\
then we  have  $\ppf{f}{}: \PPa(M) \to \PPa(N)$ mapping each subset element by element and  for example\\
$\ppf{f}{}(\{u,v,w,  \{u,v\}, \{v,\{ w, \emptyset \}\} \})  =  \{n_3, n_1,\{n_3\},\{n_3, \{n_1,\emptyset\}\} \}$
\end{example}

 \begin{lemma}[$\PPa$ preserves injections]
\label{l.PPa_prsv_injctPB}
Given  injective function $f:M \to N$ then
$\ppf{f}{}:\PPa(M) \to \PPa(N)$  is injective.
\end{lemma}
\begin{proof}
By induction\footnote{To be precise this is an Noetherian induction see Appendix~\ref{s.NoetherianInd}} over the depth of the superpower sets, so $n$ is the number of  nested parentheses:
\begin{description}
	\item[IA $(n=0)$ atomic nodes] Given $x_1,x_2 \in \PPa( M)$ with $x_1 \neq x_2$.\\
	     Since $f$ injective, $\ppf{f}{}(x_1)= f(x_1)  \neq  f(x_2) = \ppf{f}{}(x_2) $. 
	\item[IB]  Let $\ppf{f}{}:\PPa(M) \to \PPa(N)$  be injective for all sets  with at most $n$ nested  parentheses.
	\item[IS]Given  $M_1,M_2 \in \PPa(M)$ with $M_1 \neq M_2$ having both $n+1$ nested  parentheses.\\
	               Let $x\in M_1\land x \notin M_2$. \\
								 Hence $\ppf{f}{}(x) \in \ppf{f}{}(M_1)$. \\
								$x \notin M_2$ implies for all $m\in M_2$ that $x \neq m$.\\
								$x$ and $m$ have at most $n$ nested  parentheses.  \\
								$\ppf{f}{}(x) \neq \ppf{f}{}(m)$ for all $m\in M_2$ as $\ppf{f}{}$ is injective for all sets  with at most $n$ nested  parentheses. \\
								Thus
								$\ppf{f}{}(x) \notin \ppf{f}{}(M_2)$ \\ So, $\PPa(M_1)  \neq \PPa(M_2)$.
\end{description}
\end{proof}

\begin{lemma}[$\PPa$ preserves pullbacks along injective morphisms]
\label{l.PPaprsvinjPB}
\end{lemma}
\begin{proof}
Given a pullback diagram $(PB)$ and the diagram $(1)$ in  \cSets   with $g_1: C \hookrightarrow D$ injective.\\
 Pullbacks and  the super powerset functor (see Lemma~\ref{l.PPa_prsv_injctPB})  preserve injections, so \\
$\pi_B: A\hookrightarrow B$, $\PPa(\pi_B): \PPa(A) \hookrightarrow \PPa(B)$
	and $\pi_{\PPa(B)}: P \hookrightarrow \PPa(B)$ are injective.\\
$\xymatrix@=15mm{
	           A \ar@{^{(}->}[r]|{\pi_B}  \ar[d]|{\pi_C}    \ar@{}[dr]|{(PB)}
		  & B   \ar[d]|{f_1}   \\
			   C \ar@{^{(}->}[r]|{g_1} 
			&  D
		}
	$ \hfill
	$\xymatrix@=15mm{
	            P   \ar@{..>}@/^/[dr]|{\bar{h}}  \ar@{^{(}->}@/^/[drr]|{\pi_{\PPa(B)} }  \ar@/_/[ddr]|{\pi_{\PPa(C)}  }  &&\\
       &    \PPa(A) \ar@{^{(}->}[r]|{\ppf{\pi}{B}}  \ar[d]|{\ppf{f\pi}{C}}  \ar@{}[dr]|{(1)} \ar@/^/[ul]|h \ar@{}[ur]^<<<{(2)} \ar@{}[dl]|<<<{(3)} 
		  &   \PPa(B)   \ar[d]|{\ppf{f}{1}}
			\\
			 &    \PPa(C) \ar@{^{(}->}[r]|{\ppf{g}{1}}
			&   \PPa(D)
		}
	$
	\\
	Since $(PB)$ is a pullback diagram we have  $A=\{ (b,c) \mid f_1(b) = g_1(c)\}$.\\
	 $(1)$ commutes, since \PPa  is a functor.\\[2mm]
	Let $P$ be the pullback of $( \PPa(D), \ppf{f}{1},\ppf{g}{1})$, so $ \ppf{f}{1} \circ \pi_{\PPa(B)}  =  \pi_{\PPa(C)}  \circ \ppf{g}{1}$.\\
	Hence, $P= \{(B',C') \mid \ppf{f}{1} (B') = \ppf{g}{1} (C') \}   \subseteq \PPa(B) \times \PPa(C)$.\\
	Moreover, there is the unique $h: \PPa(A) \to P$ 
	s.t. $h(A') = (\ppf{\pi}{B}(A'), \ppf{f\pi}{C}(A')) $ for all $A'\subseteq A$ so that  the diagrams $(2)$ and $(3)$ commute along $h$:\\ \hspace*{\fill} 
	$\pi_{\PPa(B)} \circ h = \ppf{\pi}{B}$ and  
	$\pi_{\PPa(C)} \circ h =  \ppf{\pi}{C}$.\\[2mm]
	We define  $\bar{h}: P \to \PPa(A)$ with \\
	$$\bar{h}((X,Y)) =  \begin{cases}
	        (b,c) \; ;  \text{ if } X =b\in B , \;  Y=c\in  C \\
	       \{ (x,y) \mid x \in X\cap B, \; y \in Y\cap C, \;  f_1(x) = g_1(y)\}\\
				 \text{\hspace*{32mm}}
	\cup \bigcup_{(X',Y') \in (X-B) \times (Y-C)} \bar{h}(X',Y')  \; ; \text{ else}
	\end{cases}
	$$
	and  have:
	\begin{enumerate}
		\item 		$\bar{h}$ is well-defined since $\bar{h}(X,Y)  \in \PPa(A)$.
	  \item $(2)$ commutes along $\bar{h}$, i.e.  $\ppf{\pi}{B} \circ \bar{h}  = \pi_{\PPa(B)}(X,Y)$		
					by induction over  the number of  nested parentheses $n$:
		 \begin{description}
			 \item[IA] ($n=0$, i.e. atomic nodes): Given $ (b,c) \in P$ with
			      $b\in B$ and $c\in C$.
			$\ppf{\pi}{B} \circ \bar{h} (b,c) = \ppf{\pi}{B}(b,c) = b = \pi_{\PPa(B)}(b,c)$
			
					\item[IB] Let  be $\ppf{\pi}{B}\circ \bar{h}  (X,Y)= \pi_{\PPa(B)}(X,Y)$	 for sets  with at most $n$ nested parentheses.
					\item[IS]  Given $(\hat{X}, \hat{Y}) \in P$ with $n+1$ nested parentheses. \\
 Let  $\hat{X} =\hat{B} \cup X$ with 
$\hat{B} \subseteq B$ and $X\cap B = \emptyset$. \\
Let $\hat{Y} =\hat{C} \cup Y$ with 
$\hat{C} \subseteq C$ and $Y\cap C = \emptyset$.  \\
$X$ and $Y$ have at most $n$  nested parentheses. 
\begin{align*}
\ppf{\pi}{B}\circ  \bar{h}(\hat{X}, \hat{Y}) \\
                       &=\ppf{\pi}{B} \circ\big( \; \{ (x,y) \mid x \in \hat{B}, \;  y \in \hat{C} \;  f_1(x) = g_1(y)\} 
                                                      \cup \bigcup_{(X',Y') \in (X\times Y)} \bar{h}(X',Y')  \;  \big)\\
                       &=\{ x\mid x \in \hat{B} \land \exists y \in \hat{C}: f_1(x) = g_1(y)\} 
											                                 \cup \bigcup_{(X',Y') \in (X\times Y)} \ppf{\pi}{B} \circ \bar{h}(X',Y')   \\\\
                       &= \hat{B}  \cup \bigcup_{(X',Y') \in (X\times Y)} \ppf{\pi}{B}\circ \bar{h}(X',Y')   \\
                       &\stackrel{IB}{=}\hat{B}  \cup \bigcup_{(X',Y') \in (X\times Y)} \pi_{\PPa(B)} (X',Y')   \\
                       &= \hat{B}  \cup \bigcup_{(X',Y') \in (X\times Y)}  X'   \\
                       &= \hat{B}  \cup X  \\
											&=  \hat{X}\\
											&= \pi_{\PPa(B)}(\hat{X}, \hat{Y}) 
					\end{align*}
		 \end{description}
	\end{enumerate}

	Now we show $P \cong \PPa(A)$:
	\begin{itemize}
	   \item $h \circ  \bar{h}= id_P$,\\
		    since $\pi_{\PPa(B)}$ is injective and 
				$\pi_{\PPa(B)} \circ h \circ  \bar{h}=  \ppf{\pi}{B} \circ  \bar{h} = \pi_{\PPa(B)} \circ id_P$
		
		\item $\bar{h} \circ h = id_{\PPa(A)}$, \\
		since $\ppf{\pi}{B}$ is injective and 
		$\ppf{\pi}{B} \circ \bar{h} \circ h = \pi_{\PPa(B)}  \circ h = \ppf{\pi}{B} =\ppf{\pi}{B} \circ  id_{\PPb(A)}$.
\end{itemize}
\end{proof}

%\subsection{Coalgebraic $F$-Graph}
Based on $F$-graphs (see \cite{Jae2015,Schn99}), that is a    family  of graph categories induced by a comma category construction using a functor $F$, we can define the category of  $\PPa$-graphs.

\begin{example}[$\PPa$-Graph]\label{ex.PPa_graph}
\begin{figure}[h]
	\includegraphics[width=1.00\textwidth]{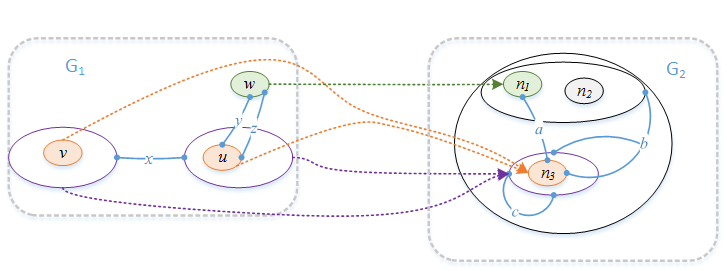}
	\caption{$\PPa$-graph morphism}
	\label{fig:ex_PPa_Graph_mor_2}
\end{figure}
In Fig.~\ref{fig:ex_PPa_Graph_mor_2} the following $\PPa$-graph are illustrated and a morphism in between.\\
	\parbox{.49\linewidth}{
$G_1=(\st_1: E_1 \to \PPa(N_1))$ \\with
 $ \begin{array}[t]{rr@{\; \mapsto\; }l}
	\st_1:  & x &    \{\{u\},\{v\}\}\\
				&	y & \{u,w\} \\
				&	z&   \{u,w\}
\end{array}
$}
	\parbox{.49\linewidth}{
$G_2=(\st_2: E_2 \to \PPa(N_2))$ \\
with $ \begin{array}[t]{rr@{\; \mapsto\; }l}
	\st_2:& a &    \{n_1,n_2\}\\
				&	b & \{ \{n_1,n_2\},  \{n_3\}, n_3\} \\
				&	c&   \{\{n_3\}\}
\end{array}
$}\\
Note, that we only have the nesting of nodes, but the nodes that contain others do not have a name themselves.
Edges are hyperedges given as a subset of the superpower set, but they may have incident nodes as well as nodes containing nodes.
\end{example}

\begin{definition}[$\PPa$-graph and the category of $\PPa$-graphs]\label{d.cr_FGraph}
The category of $\PPa$-graphs $\PPaGraph$ is given by a comma category $\PPaGraph=< Id_\cSets \downarrow \PPa>$.
\end{definition}

 $\PPa$-graph morphisms are given by mappings of the nodes and edges $f =(f_{\NN},f_{\EN}): G_1 \to G_2$ with $f_{\NN}:\NN_1 \to\NN_2$ and $f_{\EN}:\EN_1 \to \EN_2$ so that:

\begin{itemize}
	\item $\ppf{f}{\NN}\circ \node_1 = \node_2 \circ f_{\NN}$
	\item $\ppf{f}{\EN} \circ \st_1 = \st_2 \circ f_{\EN}$
\end{itemize}

\subsection{Node Recursion Based on $\PPb$}

\begin{definition}[Superpower set $\PPb$]
\label{d.ppb_set}
Given a well-founded set  $M$ and $\Pot (M) $ the power set of $M$ then we define the superpower set $\PPb(M)$
	\begin{enumerate}
	\item $\Pot(M) \subset \PPb(M)$
	\item If $M' \subset \PPb(M)$ then $M' \in \PPb(M)$.
\end{enumerate}
$\PPb(M)$ is the smallest set satisfying 1. and 2.
\end{definition}
Note that this superpower set construction $\PPb$ is well-founded sets with the ordering with respect to the number of parentheses (see Appendix~\ref{s.NoetherianInd}).
\\
The difference between  $\PPa$ and $\PPb$ is that in $\PPa$ the elements of the underlying set are elements of the superpowerset as well, so $M\subset \PPa(M)$ but $M\not\subset \PPb$ for an non-empty set $M$.

\begin{lemma}[$\PPb$ is a functor]\label{l.PPb_functor}
$\PPb:\cSets \to \cSets$ is defined for sets as in Def.~\ref{d.ppb_set} and for functions $f:M \to N$ by
$\ppf{f}{}:\PPb(M) \to \PPb(N)$ with\\
$\ppf{f}{}(\{x\})=
							\begin{cases}
	            \{f(x)\}  &\; ; x\in M\\
							\{\ppf{f}{} (x') \mid x' \in x\} &\; ; \text{ else}
  \end{cases}
$
\end{lemma}

\begin{example}[Functor $\PPb$]
Given sets $M=\{u,v,w,u',v'\}$ and $N=\{ n_1,n_2,n_3, n_4, n_5, n_6\}$  with $f:M\to N$ and
$\begin{array}[t]{cr@{\mapsto}l@{\; ; \;}r@{\mapsto}l @{\; ; \;}r@{\mapsto}l}
    f: &u & n_3  & v & n_3 & w &n_1 \\
		   &u' & n_5 & v'& n_5 \\
\end{array}
$ \\
then we  have  $\ppf{f}{}: \PPb(M) \to \PPb(N)$ with for example\\
$\ppf{f}{}(\{ \{u,v\}, \{v,\{ w, \emptyset \}\} \})  =  \{\{n_3\},\{n_3, \{n_1,\emptyset\}\} \}$
\end{example}

\begin{lemma}[$\PPb$ preserves injections]
\label{l.PPb_prsv_injct}
Given an  injective function $f:M \to N$ then
$\ppf{f}{}:\PPb(M) \to \PPb(N)$  is injective.
\end{lemma}
\begin{proof}
By induction\footnote{To be precise this is an Noetherian induction see Appendix~\ref{s.NoetherianInd}} over the depth of the superpower sets, so $n$ is the number of  nested parentheses:
\begin{description}
	\item[IA $(n=1)$ i.e. sets] Given $M_1,M_2 \subseteq M$ with $M_1 \neq M_2$.\\
	            Let $x\in M_1\land x \notin M_2$. Hence $f(x) \in f(M_1) \land  f(x) \notin f(M_2)$ as $f$ is injective.  \\
							So, $\PPb(M_1)  \neq \PPb(M_2)$.
	\item[IB]  Let $\ppf{f}{}:\PPb(M) \to \PPb(N)$  be injective for all sets  with at most $n$ nested  parentheses.
	\item[IS]  Given  $M_1,M_2 \in \PPb(M)$ with $M_1 \neq M_2$ having both $n+1$ nested  parentheses.\\
	               Let $x\in M_1\land x \notin M_2$. \\
								 Hence $\ppf{f}{}(x) \in \ppf{f}{}(M_1)$. \\
								$x \notin M_2$ implies for all $m\in M_2$ that $x \neq m$.\\
								$x$ and $m$ have at most $n$ nested  parentheses. \\
								$\ppf{f}{}(x) \neq \ppf{f}{}(m)$ for all $m\in M_2$ as $\ppf{f}{}$ is injective for all sets  with at most $n$ nested  parentheses. \\
								Thus
								$\ppf{f}{}(x) \notin \ppf{f}{}(M_2)$ \\ So, $\PPb(M_1)  \neq \PPb(M_2)$.
\end{description}
\end{proof}

\begin{lemma}[$\PPb$ preserves pullbacks along injective morphisms]
\end{lemma}
\begin{proof}
Given a pullback diagram $(PB)$ and the diagram $(1)$ in  \cSets   with $g_1: C \hookrightarrow D$ injective.\\
 Pullbacks and  the superpower set functor (see Lemma~\ref{l.PPb_prsv_injct})  preserve injections, so \\
$\pi_B: A\hookrightarrow B$, $\PPb(\pi_B): \PPb(A) \hookrightarrow \PPb(B)$
	and $\pi_{\PPb(B)}: P \hookrightarrow \PPb(B)$ are injective.\\
$\xymatrix@=15mm{
	           A \ar@{^{(}->}[r]|{\pi_B}  \ar[d]|{\pi_C}    \ar@{}[dr]|{(PB)}
		  & B   \ar[d]|{f_1}   \\
			   C \ar@{^{(}->}[r]|{g_1} 
			&  D
		}
	$ \hfill
	$\xymatrix@=15mm{
	            P   \ar@{..>}@/^/[dr]|{\bar{h}}  \ar@{^{(}->}@/^/[drr]|{\pi_{\PPb(B)} }  \ar@/_/[ddr]|{\pi_{\PPb(C)}  }  &&\\
       &    \PPb(A) \ar@{^{(}->}[r]|{\ppf{\pi}{B}}  \ar[d]|{\ppf{f\pi}{C}}  \ar@{}[dr]|{(1)} \ar@/^/[ul]|h \ar@{}[ur]^<<<{(2)} \ar@{}[dl]|<<<{(3)} 
		  &   \PPb(B)   \ar[d]|{\ppf{f}{1}}
			\\
			 &    \PPb(C) \ar@{^{(}->}[r]|{\ppf{g}{1}}
			&   \PPb(D)
		}
	$
	\\
	Since $(PB)$ is a pullback diagram we have  $A=\{ (b,c) \mid f_1(b) = g_1(c)\}$.\\
	 $(1)$ commutes, since \PPb  is a functor.\\[2mm]
	Let $P$ be the pullback of $( \PPb(D), \ppf{f}{1},\ppf{g}{1})$, so $ \ppf{f}{1} \circ \pi_{\PPb(B)}  =  \pi_{\PPb(C)}  \circ \ppf{g}{1}$.\\
	Hence, $P= \{(B',C') \mid \ppf{f}{1} (B') = \ppf{g}{1} (C') \}   \subseteq \PPb(B) \times \PPb(C)$.\\
	Moreover, there is the unique $h: \PPb(A) \to P$ 
	s.t.  $h(A') = (\ppf{\pi}{B}(A'), \ppf{f\pi}{C}(A')) $ for all $A'\subseteq A$ so that  the diagrams $(2)$ and $(3)$ commute along $h$:\\ \hspace*{\fill} 
	$\pi_{\PPb(B)} \circ h = \ppf{\pi}{B}$ and  
	$\pi_{\PPb(C)} \circ h =  \ppf{\pi}{C}$.\\[2mm]
	We define  $\bar{h}: P \to \PPb(A)$ with 
	$$\bar{h}(X,Y) = \{ (x,y) \mid x \in X\cap B, \; y \in Y\cap C, \;  f_1(x) = g_1(y)\}
	\cup \bigcup_{(X',Y') \in (X-B) \times (Y-C)} \bar{h}(X',Y')
	$$
	and  have:
	\begin{enumerate}
		\item 		$\bar{h}$ is well-defined since $\bar{h}(X,Y)  \in \PPb(A)$.
	  \item $(2)$ commutes along $\bar{h}$, i.e.  $\ppf{\pi}{B} \circ \bar{h}  = \pi_{\PPb(B)}(X,Y)$		
					by induction over  the number of  nested parentheses $n$:
		 \begin{description}
			 \item[IA] ($n=1$, i.e. sets): 
			       $\ppf{\pi}{B} \circ \bar{h} (\emptyset,\emptyset) = \ppf{\pi}{B}( \emptyset) = \emptyset  = \pi_{\PPb(B)}(\emptyset,\emptyset) )$\\
						 For $X \subseteq B$  and $Y \subseteq C$
			\begin{align*}
			    \ppf{\pi}{B} \circ \bar{h} (X,Y)\\
					                 & =  \ppf{\pi}{B}( \{ (x,y) \mid x \in X, \; y \in Y, \;  f_1(x) = g_1(y)\}  )\\
													&=  \{ x \mid x \in X \land  \exists; y \in Y : f_1(x) = g_1(y)\}) \\
													&=X \\
													& = \pi_{\PPb(B)}(X,Y)
					\end{align*}
					\item[IB] Let  be $\ppf{\pi}{B} \circ \bar{h}  (X,Y)= \pi_{\PPb(B)}(X,Y)$	 for sets  with at most $n$ nested parentheses.
					\item[IS]  Given $(\hat{X}, \hat{Y}) \in P$ with $n+1$ nested parentheses. \\
 Let  $\hat{X} =\hat{B} \cup X$ with 
$\hat{B} \subseteq B$ and $X\cap B = \emptyset$. \\
Let $\hat{Y} =\hat{C} \cup Y$ with 
$\hat{C} \subseteq C$ and $Y\cap C = \emptyset$.  \\
$X$ and $Y$ have at most $n$  nested parentheses. 
\begin{align*}
\ppf{\pi}{B} \circ  \bar{h}(\hat{X}, \hat{Y}) \\
                       &=\ppf{\pi}{B} \circ\big( \; \{ (x,y) \mid x \in \hat{B}, \;  y \in \hat{C} \;  f_1(x) = g_1(y)\} 
                                                      \cup \bigcup_{(X',Y') \in (X\times Y)} \bar{h}(X',Y')  \;  \big)\\
                       &=\{ x\mid x \in \hat{B} \land \exists y \in \hat{C}: f_1(x) = g_1(y)\} 
											                                 \cup \bigcup_{(X',Y') \in (X\times Y)} \ppf{\pi}{B} \circ \bar{h}(X',Y')   \\\\
                       &= \hat{B}  \cup \bigcup_{(X',Y') \in (X\times Y)} \ppf{\pi}{B} \circ \bar{h}(X',Y')   \\
                       &\stackrel{IB}{=}\hat{B}  \cup \bigcup_{(X',Y') \in (X\times Y)} \pi_{\PPb(B)} (X',Y')   \\
                       &= \hat{B}  \cup \bigcup_{(X',Y') \in (X\times Y)}  X'   \\
                       &= \hat{B}  \cup X  \\
											&=  \hat{X}\\
											&= \pi_{\PPb(B)}(\hat{X}, \hat{Y}) 
					\end{align*}
		 \end{description}
	\end{enumerate}

	Now we show $P \cong \PPb(A)$:
	\begin{itemize}
	   \item $h \circ  \bar{h}= id_P$,\\
		    since $\pi_{\PPb(B)}$ is injective and 
				$\pi_{\PPb(B)} \circ h \circ  \bar{h}=  \ppf{\pi}{B}  \circ  \bar{h} = \pi_{\PPb(B)} \circ id_P$
		
		\item $\bar{h} \circ h = id_{\PPb(A)}$, \\
		since $\ppf{\pi}{B}$ is injective and 
		$\ppf{\pi}{B}\circ \bar{h} \circ h = \pi_{\PPb(B)}  \circ h = \ppf{\pi}{B}= \ppf{\pi}{B} \circ  id_{\PPb(A)}$.
\end{itemize}
\end{proof}
Based on $F$-graphs (see \cite{Jae2015,Schn99}), that is a    family  of graph categories induced by a comma category construction using a functor $F$, we can define the category of coalgebraic $F$-graphs.

\begin{example}[$\PPb$-Graph]\label{ex.PPb_graph}
\begin{figure}[h]
	\includegraphics[width=1.00\textwidth]{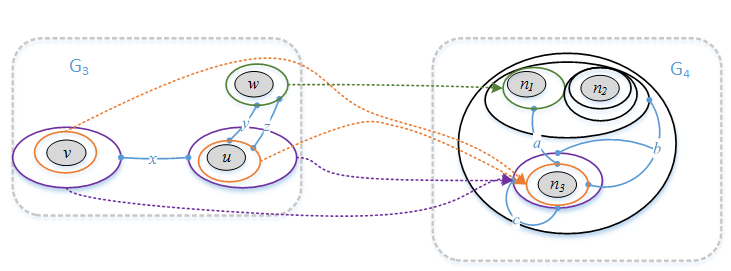}
	\caption{$\PPb$-graph morphism}
	\label{fig:ex_PPb_Graph_mor_2}
\end{figure}
In Fig.~\ref{fig:ex_PPb_Graph_mor_2} the following $\PPb$-graphs are illustrated and a morphism in between.\\
	\parbox{.49\linewidth}{
$G_3=(\st_3: E_3 \to \PPb(N_3))$ \\with
 $ \begin{array}[t]{rr@{\; \mapsto\; }l}
	c_3:  & x &    \{\{\{u\}\},\{\{v\}\}\}\\
				&	y & \{\{u\},\{w\}\} \\
				&	z&   \{\{u\},\{w\}\}
\end{array}
$}
	\parbox{.49\linewidth}{
$G_4=(\st_4: E_4 \to \PPb(N_4))$ \\
with $ \begin{array}[t]{rr@{\; \mapsto\; }l}
	\st_4:& a &    \{\{n_1\},\{n_4\}\}\\
				&	b & \{ \{n_1,n_4\},  \{n_3\}, \{\{n_3\}\}\} \\
				&	c&   \{\{\{n_3\}\}\}
\end{array}
$}
\\
Note, that we only have the recursion of nodes, but the nodes that contain others do not have a name themselves.
Edges are hyperedges given as a subset of the superpower set, but they only have neighbours that are  nodes containing nodes. They cannot have incident vertices.
\end{example}

\begin{definition}[$\PPb$-graphs and the category $\PPbGraph$]
\label{d.cr_FGraph}
The category of coalgebraic graphs $\PPbGraph$ is given by a comma category $\PPbGraph=< Id_\cSets \downarrow \PPb>$.
\end{definition}
Morphisms are given by mappings of the nodes and edges $f =(f_{\NN},f_{\EN}): G_1 \to G_2$ with $f_{\NN}:\NN_1 \to\NN_2$ and $f_{\EN}:\EN_1 \to \EN_2$ so that:
\begin{itemize}
	\item $\ppf{f}{\NN}\circ \node_1 = \node_2 \circ f_{\NN}$
	\item $\ppf{f}{\EN} \circ \st_1 = \st_2 \circ f_{\EN}$
\end{itemize}

\subsection{Node recursion based on $\Potom$}
\begin{definition}[Superpower set $\Potom$]\label{d.potom_set}
Given a set well-founded $M$ we define $\Pot^0 (M) =M$ and $\Pot^1 (M) = \Pot (M) $ the power set of $M$. Then $\Pot^{i+1} (M) = \Pot (\Pot^i (M))$\\

$$\Potom (M) = \bigcup_{i \in \Nat_0} \Pot^i (M)$$
\end{definition}
Note that this superpower set construction $\Potom$ is well-founded sets with the ordering with respect to the number of parentheses (see Appendix~\ref{s.NoetherianInd}).\\
This construction differs from the other notions, as in each subset there are only subsets that the the same depth in terms of nesting. So, for some non-empty set $M$ with $m\in M$, we have $\{m, M\} \notin \Potom(M)$ but $\{m, M\} \in \PPa(M)$  and $\{m, M\} \in \PPb(M)$.

\begin{lemma}[$\Potom$ is a functor]\label{l.Potom_functor}
$\Potom:\cSets \to \cSets$ is defined for sets as in Def.~\ref{d.potom_set} and for functions $f:M \to N$ by
$\ppf{f}{}:\Potom(M) \to \Potom(N)$ with\\
$\ppf{f}{}(\{x\})=
							\begin{cases}
	            \{f(x)\}  &\; ; x\in M\\
							\{\ppf{f}{} (x') \mid x' \in x\} &\; ; \text{ else}
  \end{cases}
$
\end{lemma}
\begin{example}[Functor $\Potom$]
Given sets $M=\{u,v,w,u',v'\}$ and $N=\{ n_1,n_2,n_3, n_4, n_5, n_6\}$  with $f:M\to N$ and
$\begin{array}[t]{cr@{\mapsto}l@{\; ; \;}r@{\mapsto}l @{\; ; \;}r@{\mapsto}l}
    f: &u & n_3  & v & n_3 & w &n_1 \\
		   &u' & n_5 & v'& n_5 \\
\end{array}
$ \\
then we  have  $\ppf{f}{}: \Potom(M) \to \Potom(N)$ with for example\\
$\ppf{f}{}(\{ \{u,v\}, \{\emptyset ,\{ w, \emptyset \}\} \})  =  \{\{n_3\},\{\emptyset, \{n_1,\emptyset\}\} \}$
\end{example}

\begin{lemma}[$\Potom$ preserves injections]
\label{l.Potom_prsv_injct}
Given  injective function $f:M \to N$ then
$\ppf{f}{}:\Potom(M) \to \Potom(N)$  is injective.
\end{lemma}
\begin{proof}
By induction\footnote{To be precise this is an Noetherian induction see Appendix~\ref{s.NoetherianInd}} over the depth of the superpower sets, so $n$ is the number of  nested parentheses:
\begin{description}
	\item[IA $(n=0)$ atomic nodes] Given $x_1,x_2 \in \Potom( M)$ with $x_1 \neq x_2$.\\
	     Since $f$ injective, $\ppf{f}{}(x_1)= f(x_1)  \neq  f(x_2) = \ppf{f}{}(x_2) $. 
	\item[IB]  Let $\ppf{f}{}:\Potom(M) \to \Potom(N)$  be injective for all sets  with at most $n$ nested  parentheses.
	\item[IS]  Given  $M_1,M_2 \in \Potom(M)$ with $M_1 \neq M_2$ having both $n+1$ nested  parentheses.\\
	               Let $x\in M_1\land x \notin M_2$. \\
								 Hence $\ppf{f}{}(x) \in \ppf{f}{}(M_1)$. \\
								$x \notin M_2$ implies for all $m\in M_2$ that $x \neq m$.\\
								$x$ and $m$ have at most $n$ nested  parentheses. \\
								$\ppf{f}{}(x) \neq \ppf{f}{}(m)$ for all $m\in M_2$ as $\ppf{f}{}$ is injective for all sets  with at most $n$ nested  parentheses. \\
								Thus
								$\ppf{f}{}(x) \notin \ppf{f}{}(M_2)$ \\ So, $\Potom(M_1)  \neq \Potom(M_2)$.
\end{description}
\end{proof}

\begin{lemma}[$\Potom$ preserves pullbacks along injective morphisms]
\end{lemma}
\begin{proof}
Given a pullback diagram $(PB)$ and the diagram $(1)$ in  \cSets   with $g_1: C \hookrightarrow D$ injective.\\
 Pullbacks and  the superpower set functor (see Lemma~\ref{l.Potom_prsv_injct})  preserve injections, so \\
$\pi_B: A\hookrightarrow B$, $\Potom(\pi_B): \Potom(A) \hookrightarrow \Potom(B)$
	and $\pi_{\Potom(B)}: P \hookrightarrow \Potom(B)$ are injective.\\
$\xymatrix@=15mm{
	           A \ar@{^{(}->}[r]|{\pi_B}  \ar[d]|{\pi_C}    \ar@{}[dr]|{(PB)}
		  & B   \ar[d]|{f_1}   \\
			   C \ar@{^{(}->}[r]|{g_1} 
			&  D
		}
	$ \hfill
	$\xymatrix@=15mm{
	            P   \ar@{..>}@/^/[dr]|{\bar{h}}  \ar@{^{(}->}@/^/[drr]|{\pi_{\Potom(B)} }  \ar@/_/[ddr]|{\pi_{\Potom(C)}  }  &&\\
       &    \Potom(A) \ar@{^{(}->}[r]|{\ppf{\pi}{B}}  \ar[d]|{\ppf{f\pi}{C}}  \ar@{}[dr]|{(1)} \ar@/^/[ul]|h \ar@{}[ur]^<<<{(2)} \ar@{}[dl]|<<<{(3)} 
		  &   \Potom(B)   \ar[d]|{\ppf{f}{1}}
			\\
			 &    \Potom(C) \ar@{^{(}->}[r]|{\ppf{g}{1}}
			&   \Potom(D)
		}
	$
	\\
	Since $(PB)$ is a pullback diagram we have  $A=\{ (b,c) \mid f_1(b) = g_1(c)\}$.\\
	 $(1)$ commutes, since \Potom  is a functor.\\[2mm]
	Let $P$ be the pullback of $( \Potom(D), \ppf{f}{1},\ppf{g}{1})$, so $ \ppf{f}{1} \circ \pi_{\Potom(B)}  =  \pi_{\Potom(C)}  \circ \ppf{g}{1}$.\\
	Hence, $P= \{(B',C') \mid \ppf{f}{1} (B') = \ppf{g}{1} (C') \}   \subseteq \Potom(B) \times \Potom(C)$.\\
	Moreover, there is the unique $h: \Potom(A) \to P$ 
	s.t. $h(A') = (\ppf{\pi}{B}(A'), \ppf{f\pi}{C}(A')) $ for all $A'\subseteq A$ so that  the diagrams $(2)$ and $(3)$ commute along $h$:\\ \hspace*{\fill} 
	$\pi_{\Potom(B)} \circ h = \ppf{\pi}{B}$ and  
	$\pi_{\Potom(C)} \circ h =  \ppf{\pi}{C}$.\\[2mm]
	We define  $\bar{h}: P \to \Potom(A)$ with 
	\\
	$\bar{h}((X,Y)) =  \begin{cases}
	        (b,c) \; ;  \text{ if } X =b\in B , \;  Y=c\in  C \\
	       \{ (x,y) \mid x \in X\cap B, \; y \in Y\cap C, \;  f_1(x) = g_1(y)\}\\
				 \text{\hspace*{32mm}}
	\cup \bigcup_{(X',Y') \in (X-B) \times (Y-C)} \bar{h}(X',Y')  \; ; \text{ else}
	\end{cases}
	$
	 \\
	and  have:
	\begin{enumerate}
		\item 		$\bar{h}$ is well-defined since $\bar{h}(X,Y)  \in \Potom(A)$.
	  \item $(2)$ commutes along $\bar{h}$, i.e.  $\ppf{\pi}{B} \circ \bar{h}  = \pi_{\Potom(B)}(X,Y)$		
					by induction over  the number of  nested parentheses $n$:
		 \begin{description}
		
			 \item[IA] ($n=0$, i.e. atomic nodes): Given $ (b,c) \in P$ with
			      $b\in B$ and $c\in C$.
			$\ppf{\pi}{B} \circ \bar{h} (b,c) = \ppf{\pi}{B}(b,c) = b = \pi_{\Potom(B)}(b,c)$

					\item[IB] Let  be $\ppf{\pi}{B} \circ \bar{h}  (X,Y)= \pi_{\Potom(B)}(X,Y)$	 for sets  with at most $n$ nested parentheses.
					\item[IS]  Given $(\hat{X}, \hat{Y}) \in P$ with $n+1$ nested parentheses. \\
 Let  $\hat{X} =\hat{B} \cup X$ with 
$\hat{B} \subseteq B$ and $X\cap B = \emptyset$. \\
Let $\hat{Y} =\hat{C} \cup Y$ with 
$\hat{C} \subseteq C$ and $Y\cap C = \emptyset$.  \\
$X$ and $Y$ have at most $n$  nested parentheses. 
\begin{align*}
\ppf{\pi}{B} \circ  \bar{h}(\hat{X}, \hat{Y}) \\
                       &=\ppf{\pi}{B} \circ\big( \; \{ (x,y) \mid x \in \hat{B}, \;  y \in \hat{C} \;  f_1(x) = g_1(y)\} 
                                                      \cup \bigcup_{(X',Y') \in (X\times Y)} \bar{h}(X',Y')  \;  \big)\\
                       &=\{ x\mid x \in \hat{B} \land \exists y \in \hat{C}: f_1(x) = g_1(y)\} 
											                                 \cup \bigcup_{(X',Y') \in (X\times Y)} \ppf{\pi}{B} \circ \bar{h}(X',Y')   \\\\
                       &= \hat{B}  \cup \bigcup_{(X',Y') \in (X\times Y)} \ppf{\pi}{B} \circ \bar{h}(X',Y')   \\
                       &\stackrel{IB}{=}\hat{B}  \cup \bigcup_{(X',Y') \in (X\times Y)} \pi_{\Potom(B)} (X',Y')   \\
                       &= \hat{B}  \cup \bigcup_{(X',Y') \in (X\times Y)}  X'   \\
                       &= \hat{B}  \cup X  \\
											&=  \hat{X}\\
											&= \pi_{\Potom(B)}(\hat{X}, \hat{Y}) 
					\end{align*}
		 \end{description}
	\end{enumerate}

	Now we show $P \cong \Potom(A)$:
	\begin{itemize}
	   \item $h \circ  \bar{h}= id_P$,\\
		    since $\pi_{\Potom(B)}$ is injective and 
				$\pi_{\Potom(B)} \circ h \circ  \bar{h}=  \ppf{\pi}{B}  \circ  \bar{h} = \pi_{\Potom(B)} \circ id_P$
		
		\item $\bar{h} \circ h = id_{\Potom(A)}$, \\
		since $\ppf{\pi}{B}$ is injective and 
		$\ppf{\pi}{B}\circ \bar{h} \circ h = \pi_{\Potom(B)}  \circ h = \ppf{\pi}{B}= \ppf{\pi}{B} \circ  id_{\Potom(A)}$.
\end{itemize}
\end{proof}
Based on $F$-graphs (see \cite{Jae2015,Schn99}), that is a    family  of graph categories induced by a comma category construction using a functor $F$, we can define the category of coalgebraic $F$-graphs.

\begin{example}[Coalgebraic $F$-Graph based on $\Potom$]\label{ex.Potom_graph}
\begin{figure}[h]
	\includegraphics[width=1.00\textwidth]{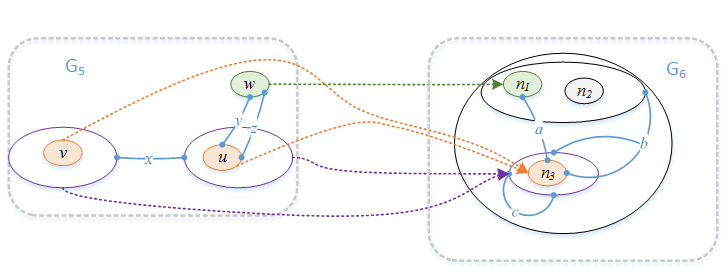}
	\caption{$\Potom$-graph morphism}
	\label{fig:ex_Potom_Graph_mor_2}
\end{figure}
In Fig.~\ref{fig:ex_Potom_Graph_mor_2} the following coalgebraic $F$-graphs are illustrated and a morphism in between.\\
	\parbox{.49\linewidth}{
$G_5=(\st_5: E_5 \to \Potom(N_5))$ \\with
 $ \begin{array}[t]{rr@{\; \mapsto\; }l}
	\st_5:  & x &    \{\{\{u\}\},\{\{v\}\}\}\\
				&	y & \{\{u\},\{w\}\} \\
				&	z&   \{\{u\},\{w\}\}
\end{array}
$}
	\parbox{.49\linewidth}{
$G_6=(\st_6: E_6 \to \Potom(N_6))$ \\
with $ \begin{array}[t]{rr@{\; \mapsto\; }l}
	\st_6:& a &    \{\{n_1\},\{n_2\}\}\\
				&	b & \{ \{n_1,n_2\},  \{n_3\}, n_3\} \\
				&	c&   \{\{n_3\}\}
\end{array}
$}
\\
Note, that we only have the recursion of nodes, but the nodes that contain others do not have a name themselves.
Edges are hyperedges given as a subset of the superpower set, but they cannot have incident vertices.
\end{example}

\begin{definition}[$\Potom$-Graph and the category of coalgebraic graphs]\label{d.cr_FGraph}
The category of $\Potom$-graphs $\PotomGraph$ is given by a comma category $\PotomGraph =< Id_\cSets \downarrow \Potom>$.
\end{definition}
 $\Potom$-graph morphisms are given by mappings of the nodes and edges $f =(f_{\NN},f_{\EN}): G_1 \to G_2$ with $f_{\NN}:\NN_1 \to\NN_2$ and $f_{\EN}:\EN_1 \to \EN_2$ so that:
\begin{itemize}
	\item $\ppf{f}{\NN}\circ \node_1 = \node_2 \circ f_{\NN}$
	\item $\ppf{f}{\EN} \circ \st_1 = \st_2 \circ f_{\EN}$
\end{itemize}

\subsection{Differences between $\PPa$, $\PPb$ and $\Potom$}\label{ss.diff}
We have for any set $S$ that $\PPb(M) \subseteq \PPa(M)$ and $\Potom(M) \subseteq \PPa(M)$.
The graphs in examples Ex.~\ref{Ex.PPa_graph}, Ex.~\ref{ex.PPb_graph},  and Ex.~\ref{ex.Potom_graph}  and the superpower set functors are listed in the table below stating which graph can be constructed using which functor.

	\parbox{.34\linewidth}{
	\begin{tabular}{r||c|c|c|}
		                & $\PPb$    & $\PPa$   & $\Potom$ \\ \hline \hline
					$G1$ &   yes                         & yes    & yes  \\
					$G2$ &   yes                         & yes    & no  see \ref{i}  \\
					$G3$ &    no see  \ref{ii}  & yes    & yes\\
					$G4$  &   no see \ref{iii}  & yes    & no  see \ref{iv}   \\
					$G5$ &   no see  \ref{ii}   & yes    & yes  \\
					$G6$ &   no see  \ref{v}    & yes    & yes  \\
	\end{tabular}}
	\parbox{.65\linewidth}{\begin{enumerate}
	\item \label{i} because $\{ \{n_1\} ,\{\{n_2\}\} \} \notin \Potom(\{n_1,n_2,n_3\})$
	\item \label{ii}   because $u, w \notin \PPb(\{u,v,w\})$
	\item \label{iii} because $\{ n_1, \{n_2\}  \} \notin \PPb(\{n_1,n_2,n_3\})$
	\item \label{iv} because $\{ n_1 ,\{n_2\} \} \notin \Potom(\{n_1,n_2,n_3\})$
	\item \label{v} because $n_1 \notin \PPb(\{n_1,n_2,n_3\})$
	\end{enumerate}}

\section{$\M$-adhesive Categories of $F$-Coalgebras}
\label{s.VKcoalg}
% VKcoalg.tex

A endofunctor $F:\cSets \to \cSets$  gives rise the  category of coalgebras $\cSetsF{F}$
with $M \stackrel{\alpha_M}{\longleftarrow} F(M)$ -- also denoted by $(M,\alpha_M)$ -- being the objects and morphisms $f:(M,\alpha_M) \to (N,\alpha_N)$ -- called $F$-homomorphism -- so that $(1)$ commutes in \cSets (see \cite{Ru00} ):
$$
   \xymatrix{
	     M \ar[r]^{\alpha_M}  \ar[d]|f  \ar@{}[dr]|{(1)}
	&  F(M)      \ar[d]|{F(f)} \\
	     N \ar[r]^{\alpha_N}
	&  F(N)
	}
$$

\begin{lemma}[Pullbacks along injections in \cSetsF{F}]
Given a functor $F: \cSets \to \cSets$  that preserves pullbacks along an  injective morphism,
then \cSetsF{F} has pullbacks along an  injective F-homomorphism.
\end{lemma}

\begin{proof}
Given $(B,\alpha_B) \injct{f} (D,\alpha_D)  \von{g} (C,\alpha_C)$. Then we have 
(PB1) in \cSets below. This results in $(PB2)$ since $F$ preserves pullbacks along an  injective morphism. $\alpha_A$ is the unique induced morphism for 
$ F(g) \circ ( \alpha_C \circ \pi_C)= F(f) \circ ( \alpha_B \circ \pi_B)$. So, $(1)$ commutes.\\[2mm]
$(A,\alpha_A)$ is pullback in \cSetsF{F} because we have the diagrams below\\
   	\parbox{.6\linewidth}{
	in \cSets: \\
		$
	   \xymatrix@C=6mm{
		      &   X \ar[rrr]^{\alpha_X}  \ar@/_/[ddl]_{g'}\ar@/^/[dddr]|{f'} \ar[d]|h
			&&&  F(X)    \ar@/_/[ddl]_(.3){F(g')}\ar@/^/[dddr]|{F(f')}   \ar[d]|{F(h)}\\
		      &  A   \ar[rrr]^{\alpha_A}  \ar[dl]|{\pi_B}    \ar@{^{(}->}[ddr]|(.6){\pi_C}    \ar@{}[ddd]|{(PB1)}
			&&&  F(A)   \ar[dl]|{F(\pi_B)}    \ar@{^{(}->}[ddr]|{F(\pi_C)}      \ar@{}[ddd]|{(PB2)}\\
					    B  \ar[rrr]^(.3){\alpha_B}  \ar@{^{(}->}[ddr]|f          
			&&& F(B) \ar@{^{(}->}[ddr]|(.6){F(f)} \\
				&&  C   \ar[rrr]^(.3){\alpha_C} \ar[dl]|g
			&&&  F(C)   \ar[dl]|{F(g)}\\
					 & D		  \ar[rrr]^{\alpha_D}
			&&& F(D)
		}
	 $
	}
	\parbox{.4\linewidth}{
	in \cSetsF{F}:\\
		$
	   \xymatrix@C=6mm{
		    &  (X,\alpha_X) \ar@/_/[ddl]_{g'}\ar@/^/[dddr]|{f'}  \ar[d]|h\\
		    & (A,\alpha_A)  \ar[dl]|{\pi_B}    \ar@{^{(}->}[ddr]|(.6){\pi_C}    \ar@{}[ddd]|{(1)} \\
				    (B,\alpha_B)   \ar@{^{(}->}[ddr]|f    \\
				&& (C,\alpha_C)  \ar[dl]|g\\
				& (D,\alpha_D)
		}
		$
		}
		 \\
		The comparison object $(X,\alpha_X)$ with $f\circ g' = g \circ f'$ in $\cSetsF{F}$ leads in $\cSets$ to the induced 
		morphisms  $h:X\to A$ and $F(h) : F(X) \to F(A)$ commuting the corresponding triangles. As $F(A)$ is pullback in $\cSets$, we have $F(h) \circ \alpha_X =  \alpha_A \circ h$.
		Hence, $h$ is the induced morphism in \cSetsF{F} as well.
\end{proof}

\begin{corollary}[Pullbacks along injections  in \cSetsF{F}]
Given a functor $F: \cSets \to \cSets$  that preserves pullbacks along an  injective morphism,
then \cSetsF{F} has pullbacks along an  injective F-homomorphism.
\end{corollary}

\subsubsection{Concerning the Vertical Weak VK Square}

\begin{definition}[Class  of monomorphisms $\M$]
Let \M be  a  class  of monomorphisms in \cSets that is PO-PB-compatible, that is:
\begin{enumerate}
	\item Pushouts along \M-morphisms exist and \M is stable under pushouts.
\item Pullbacks along \M-morphisms exist and \M is stable under pullbacks.
\item \M contains all identities and is closed under composition.
\end{enumerate}
\end{definition}
According to Prop.~4.7 in \cite{Ru00}  if $f:M\to N$ is injective in \cSets  then $f$ is  an $F$-monomorphism in $\cSetsF{F}$. Obviously the class of all injective functions  $\M_F=\{(A,\alpha_A) \injct{f} (B,\alpha_B) \mid f \text{ is injective in \cSets } \}$ 
is  PO-PB-compatible.

\begin{theorem}[$(\cSetsF{F},\M_F)$ is an $\M$-Adhesive Category]
If $F$ preserves pullbacks along injective morphisms, then $(\cSetsF{F},\M_F)$ is an $\M$-adhesive category.
\end{theorem}

\begin{proof}
 $(\cSetsF{F},\M_F)$ is an $\M$-adhesive category:
 \begin{enumerate}
	 \item Pushouts in $\cSetsF{F}$ along $m\in \M_F$  exist, since 
             $\cSetsF{F}$ is finitely cocomplete (Thm 4.2 \cite{Ru00}) for arbitrary $F:\cSets \to \cSets$.
\item and they are vertical weak VK squares, i.e.  for all commutative cubes $(2)$  where all vertical morphisms
$a,b,c,d$ are in \M with the given pushout $(1)$ in the
bottom and the back squares being pullbacks. \\Then following holds :\\
The top square is a pushout if and only if the front squares are pullbacks.

  $
   \xymatrix@=3mm{
   			&\\
                       (A,\alpha_A)  \ar[rr]|{m\in\M} \ar[dd]|{f}    
	         && (B,\alpha_B)                 \ar[dd]|{g}     \\ 
	            & (1)  \\
		       (C,\alpha_C)  \ar[rr]|{n}
		    && (D,\alpha_D)
		 }
     $ \hfill
     $
    \xymatrix@=3mm{
                     &&&    (A',\alpha'_A)  \ar[ddd]|<<<<<{a}    \ar[dlll]|{f'}    \ar[drr]|{m'}     
		         &&  (2) \\
		               (C',\alpha'_C)     \ar[ddd]|{c}                    \ar[drr]|{n'}
		   &&&&& (B',\alpha'_B)     \ar[ddd]|{b}     \ar[dlll]|<<<<<{g'}          
		           &&                                                                \\
		          &&  (D',\alpha'_D)   \ar[ddd]|<<<<<{d}                                        \\
                        &&& (A,\alpha_A)                       \ar[dlll]|>>>>>{f}     \ar[drr]|{m}           \\
		                (C,\alpha_C)                                     \ar[drr]|{n}
		   &&&&& (B,\alpha_B)                       \ar[dlll]|{g}                       
		           &&                                                              \\
		          &&  (D,\alpha_D)
		   } 	        	 
  $	
	
	Since (finite) colimits and pullbacks along $\M$-morphisms are constructed on the underlying set,  square $ (1)$ and the VK-cube are
	given  for the   underlying sets in \cSets as well.
	\begin{description}
		\item[$\Rightarrow$] If in \cSetsF{F}  $(B',\alpha'_B)  \to  (D',\alpha'_D)  \gets  (C',\alpha'_C) $ is pushout over 
		 $(B',\alpha'_B)  \gets  (A',\alpha'_A) \to   (C',\alpha'_C) $  then
	$B' \to  D' \gets  C'$ is pushout over 
	 $B' \gets  A' \to   C'$ in \cSets 
	as pushout are constructed on the underlying sets. As \cSets together with the class of injective morphisms is an adhesive category (see Thm 4.6 in \cite{FAGT}), we have the  front squares are pullbacks in $\cSets$. As the vertical morphisms are injective
	the front squares are pullbacks in \cSetsF{F} provided that $F$ preserves pullbacks along injections.
	\item[$\Leftarrow$] Let 	the front squares be pullbacks in $\cSetsF{F}$, then the squares of the underlying sets are pullbacks in $\cSets$.  So the top square in pushout in \cSets and hence in $\cSetsF{F}$. 
\end{description}

	\end{enumerate}
\end{proof}

The same holds for many-sorted coalgebras over sets. 
Graphs with undirected edges can be considered as many sorted coalgebras using the functor $F:\cSets \times \cSets \to \cSets \times \cSets$ with
$F(N,E) = (N,E) \stackrel{(\finMor,<s,t>)}{\longrightarrow} (\finObj, N \times N) $ where $\finObj$ is the final object and $\finMor$ the corresponding final morphism, see e.g. \cite{Ru00}.
\begin{corollary}
If $F:\cSets \times \cSets \to \cSets \times \cSets$ preserves pullbacks along injective morphisms, then $(\cSets \to \cSets)_F,\M_F)$ is an $\M$-adhesive category for    a  class  of monomorphisms \M in \cSets that is PO-PB-compatible.
\end{corollary}

The following corollary allows $\M$ transformation systems for various dynamic systems based on $F$-coalgebras of functors the preserve pullbacks along injective morphisms.

\begin{corollary}[$\M$ transformation systems  for $F$-coalgebras]
We obtain 
\begin{itemize}
	\item $\M$-transformation systems for finitely branching non-deterministic transition systems  
	$\cSetsF{\Potfin}$, where  $(Q, \alpha_Q:Q \to \Potfin(Q)$ as finite power set functor $ \Potfin$ preserves pullbacks along injective morphisms.	
	\item $\M$-transformation systems for infinite binary trees $\cSetsF{A \times \_\times \_ }$ over an alphabet $A$ with since the product functor preserves limits.
 \item $\M$-transformation systems for  labelled transition systems over a signature $\Sigma$  with $\cSetsF{\Pot(\Sigma \times \_)}$, since the composition  preserves pullback-preservation.
\end{itemize}
\end{corollary}
\begin{example}[Transformation of a finitely branching non-deterministic transition system]
Given the transition system $(K,\alpha_K)$ with $K=\Nat$ and 
$\alpha_K(n) =\{2n+1, 2n+2\}$  that is a full  infinite binary tree. 
The rule is $L\gets K\to R$ and te morphisms are all set inclusions. The application of the rule to $Q$ leads to the transformation step given in as illustrated in Fig.~\ref{fig:ex_Coalg_fbTS_trafo}.
\begin{figure}[h]
	\centering
		\includegraphics[width=.8\linewidth]{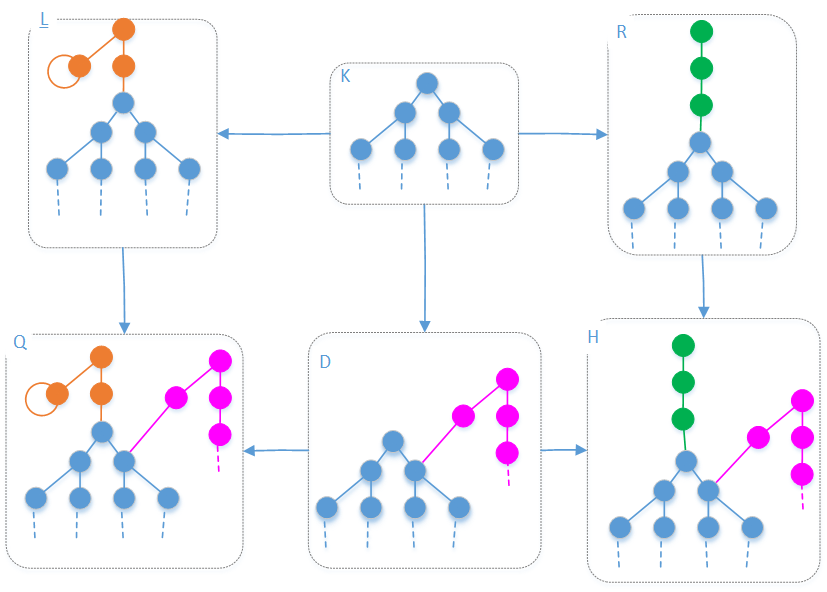}
	\caption{Transformation step in $Coalg_{\Potfin}$}
	\label{fig:ex_Coalg_fbTS_trafo}
\end{figure}
\end{example}

Although \cite{Kahl14} has already approached the coalgebraic representation of DPO-transformations this approach is far more general as its considers arbitrary coalgebras based on functors preserving pullbacks.

\subsection{Endofunctors  and Pullbacks}
\label{ss.endoPB}

A functor $F :\cC \to \cD $ is called (weak) pullback preserving if it maps (weak)
pullback squares to (weak) pullback squares: that is, if $F$ applied a (weak) pullback square
in the category $\cC$ forms a (weak) pullback square in $\cD$. (Def. 4.2.1(ii) \cite{Jac16}). 
This obviously implies that $F$ preserves pullbacks weakly, i.e.  it maps
pullback squares to (weak) pullback squares (\cite{Ru00,AD05}).
Many coalegebraic results require that $F$  preserves pullbacks weakly 
This comprises the power-set functors, and  arbitrary product, coproduct, power or composite of functors 
weakly preserving pullbacks (\cite{AD05}). 
\cite{Jac16} states that a weak pullback preserving functor
preserves (ordinary) pullbacks of monos (exercise 4.2.5  in \cite{Jac16}).
So, we have 
\begin{itemize}
	\item $F$ is (weak) pullback preserving $\Longrightarrow$ $F$ preserves pullbacks weakly.
	\item  $F$ is (weak) pullback preserving $\Longrightarrow$ $F$ preserves pullback along monomorphims.
\end{itemize}
Seemingly, all three notions hold for the power-set functors, and  arbitrary product, coproduct, power or composite of functors.
But it remains an open question if these notions are equivalent at least in $\cSets$ or in arbitrary categories.

\section{Edge Recursion}
\label{s.edgerec}
%edgerec.tex

In this section we investigate the  recursion of  edges, yielding nested edges where neighbours of edges can again be edges, as in Fig.~\ref{fig:ex1}
In \cite{Ru00} it is shown that graphs with undirected edges can be considered as many sorted coalgebras using the functor $F:\cSets \times \cSets \to \cSets \times \cSets$ with
$F(V,E) = (V,E) \nach{(\finMor,<s,t>)} (\finObj, V \times V) $ where $\finObj$ is the final object and $\finMor$ the corresponding final morphism.
In \cite{Jae2015} the notion of $F$-graphs based on comma-categories is investigated and in \cite{Jae2016}
extended to  coalgebras. The functors  investigated in \cite{Jae2016} are
the product,  the coproduct and several powerset functors.
\\
This can be extended to various types of nested edges.
\begin{definition}[Nested  hyperedges]
Given a set of nodes $\NN$  and a set of edges $\EN$ and a function yielding the neighbours 
$\st:\EN \to  \Pot (V \uplus \EN)$. \\
Then the category of coalgebras  $\coalg{F_1}$ over  $F_1: \cSets \times \cSets \to \cSets \times \cSets$  with 
$F_1(V,\EN) = (\finObj, \Pot (V\uplus \EN)$ yields the category of graphs with nested hyperedges.\\
The class $\M$ is given by the class of pairs of injective morphisms $<f_{\NN},f_{\EN}>$.
\end{definition}

\begin{lemma}[$(\coalg{F_1},\M)$ is an $\M$-adhesive category]
\end{lemma}
\begin{proof}
 $F$ preserves pullbacks along monomorphisms, as the first component is a pullback of the final object in \cSets and
the powerset functor preserves pullback of monos (see Lemma~\ref{l.Pot_PB_along_injc}) and the coproduct functor as well
(see Lemma~\ref{l.coprod_PB}).
\end{proof}

\begin{corollary}[Nested  undirected edges]
Given a set of nodes $\NN$  and a set of edge names $\EN$ and a function yielding the neighbours 
$\st:\EN \to  \PotDir(V \uplus \EN)$. 
Then the category of coalgebras  $\coalg{F_2}$ over  $F_2: \cSets \times \cSets \to \cSets \times \cSets$  with 
$F_2(\NN,\EN) = (\finObj, \PotDir (\NN\uplus \EN)$ yields the category of graphs with nested undirected edges.

The class $\M$ is given by the class of pairs of injective morphisms $<f_{\NN},f_{\EN}>$. and  $(\coalg{F_2},\M)$ is an $\M$-adhesive category.
\end{corollary}

\begin{corollary}[Nested directed edges]
Given a set of vertices $\NN$  and a set of edge names $\EN$ and a function yielding the neighbours 
$\st:\EN \to  (\NN \uplus \EN) \times  (V \uplus \EN)$.
Then the category of coalgebras  $\coalg{F_3}$ over  $F_3: \cSets \times \cSets \to \cSets \times \cSets$  with 
$F_3(\NN,\EN) = (\finObj, \NN\uplus \EN) \times  (\NN \uplus \EN) $ yields the category of graphs with nested directed edges.
\\
The class $\M$ is given by the class of pairs of injective morphisms $<f_{\NN},f_{\EN}>$ and  $(\coalg{F_3},\M)$ is an $\M$-adhesive category.
\end{corollary}

And again, these edge concepts can be mixed as well see \cite{Jae2015}.

\section{Coalgebraic Graphs}
\label{s.recgraph}
%recgraph.tex
Based on the above reached results we can now define graph where nodes and edges are nested.

\begin{example}[Nested nodes]\label{ex.nestednodes}
Nested nodes can be constructed using the  coalgebra $\coalg{\PPa}$ based on the superpower set functor $\PPa$.
Given a set $N$  the function  $\node:\NN \to \PPa(\NN)$ gives the nodes contained in a given node.
This function yields an $\M$-adhesive category; the category of coalgebras  $\coalg{\PPa}$ over  $\PPa: \cSets \to \cSets $  with
the class $\M$  of injective morphisms.\\
The nesting of nodes  can also be defined allowing the different kinds of nesting using some  
functor $\funF:\cSets \to \cSets$, so  we have the contains function $\node:\NN \to \funF(\NN)$. This  yields an $\M$-adhesive category where $\funG$ may be one of the (super-)power functors, e.g. $\Pot$, $\PotDir$, $\PPa$ or $\Potom$ or any other functor preserving pullbacks of injections. 
\end{example}

To obtain coalgebraic graphs as given in Sect.~\ref{s.motiv} we  construct coalgebraic graphs as $G= (\NN,\EN,\node:\NN \to \funF(\NN,\EN), \st: \EN \to \funG(\NN,\EN)$. These graphs can be considered to be an coalgebra over $\rGF: \cSets\times\cSets \to \cSets \times \cSets$
with  $\rGF(\NN,\EN) =(\funF(\NN,\EN), \funG(\NN,\EN))$.

\begin{definition}[Coalgebraic graph]\label{d.rGF}
	 Let the coalgebraic graph functor $\rGF: \cSets\times\cSets \to \cSets \times \cSets$ be given 
with $\rGF(\NN,\EN) =(\funF(\NN,\EN), \funG(\NN,\EN))$  where $\rGF$ preserves pullbacks along monomorphisms  provided $\funF:  \cSets\times\cSets \to \cSets $ and $\funG:  \cSets\times\cSets \to \cSets $ preserve injections and pullbacks along injective morphisms, then $\coalg{\rGF}$ is the  category of the corresponding coalgebraic graphs.
\end{definition}
According to Lemma~\ref{l.rGF_PB} in Appendix~\ref{s.misc}  the functor $\rGF$ preserves pullbacks along monomorphisms., so the have:
\begin{corollary}[Coalgebraic  graphs yield an $\M$-adhesive category.]
The class $\M$ is given by the class of pairs of injective morphisms $<f_{\NN},f_{\EN}>$ and  $(\coalg{\rGF},\M)$ is an $\M$-adhesive category.
\end{corollary}
The definition of coalgebraic graphs is chosen to be quite open and comprises the usual graph types, as (un-) directed  and (hyper-) graphs as well as various hierarchical graphs (see Subsect.~\ref{ss.exHierGra}).
\\
Obviously, in non-hierarchical graph types the contains function is superfluous, so below it is given by the final morphism $\finMor$.
\begin{itemize}

\item Undirected graphs are given by $\funF= \funFinal$ and $\funG =  \prod \circ (\PotDir \times \funFinal)$, so the objects in  $\coalg{\rGF}$ are given by 
			$\xymatrix{
			(\NN,\EN)  \ar@<1mm>[r]^(.4){\node} \ar@<-1mm>[r]_(.4){\st} &  (\funFinal, \PotDir(\NN))
			}
			$.
			
			Since $\funG(\NN,\EN) =  \prod \circ (\PotDir\times \funFinal)(\NN,\EN) =\prod  ( \PotDir(\NN) , \finObj )\cong \PotDir(\NN)$, an undirected graph is given by $G=(\NN,\EN, \node,\st)$ with $\node = \finMor$ and $\st: \EN \to \PotDir(\NN)$.
			
	\item Directed graphs can be given by $\funF= \funFinal$ and $\funG = \prod \circ \funX^2\times \funFinal$, so we have 
			$\xymatrix{
			(\NN,\EN)  \ar@<1mm>[r]^(.4){\node} \ar@<-1mm>[r]_(.4){\st} &  (\funFinal, \NN \times \NN)
			}
			$ since $\funG(\NN,\EN) =   \prod \circ \funX^2\times \funFinal (\NN,\EN) = \NN \times \NN$. So, a directed graph is an object in the category $\coalg{\rGF}$ given by $G=(\NN,\EN, \node,\st)$ with $\node = \finMor$ and $\st: \EN \to \NN \times \NN$.\\
			It would be interesting to know whether the evolving categories correspond to the usual ones, e.g. is the category $\coalg{\rGF}$ isomorphic to the usual category of graphs (as in \cite{FAGT}) given as a functor category $[\mathcal{S},\cSets]$ for the schema category 
$\mathcal{S} = \xymatrix{ \bullet \ar@<1mm>[r]\ar@<-1mm>[r] & \bullet}$.

    \item Classical hypergraphs, where edges are attached to a set of nodes,  are given by $\funF=\funFinal$ and $\funG=\Pot \times \funFinal$.
		\item Hypergraphs as in hyperedge replacement \cite{DKH97}, where edges are attached to a string of nodes, is hence given by $\funF=\funFinal$ and $\funG=(\_)^* \times \funFinal$ with $(\_)^* $ the free monoid functor.
		\\
		The definition in \cite{FAGT} uses the indexed comma categories, hence the relation between
		$\mathit{ComCat}(Id_\cSets, (\_)^* ,\{1,2\})$  and the evolving coalgebra $\coalg{\rGF}$ needs to be investigated. 
		\item Place-Transition nets can be considered to be objects in 
		$\coalg{\rGF}$, with $\funF= \funFinal$ and $\funG =( ( (\_)^* \times  (\_)^*) \circ \funX^2) \times \funFinal$. The extension of the hierarchy concepts in this contribution to Petri nets  is probably worth exploring.
\end{itemize}
Since all the involved functors preserve pullbacks of injective morphisms (see Subsect.~\ref{s.misc}), we immediately obtain $\M$-adhesive categories. \\

Subsequently, we omit the final functor $\funFinal$ for better readability.

\subsection{Examples of Hierarchical Graphs}
\label{ss.exHierGra}
In the following we relate the notion developed above to concepts of hierarchical graphs in the literature  using both concepts , comma category and colagbebra.

\begin{enumerate}
	\item The comma category $<Id_\cSets\downarrow \PPa>$ as used in Ex.~\ref{ex.PPa_graph}  with is an $\M$-adhesive category because of the  comma-category construction (see Theorem 4.15 (construction of (weak) adhesive HLR categories in \cite{FAGT}) and  $\PPa$ preserving pullbacks of injections. It yields hierarchical graphs with hyperedges between nodes and containers of nodes, but containers do not have an explicit name. 
	\item Combining the nested nodes based on the superpower set functor $\PPa$  as in Ex.~\ref{ex.nestednodes} with usual edges concepts leads to various types of hierarchical graphs and is closely related to hierarchical graphs in the sense of \cite{BKK05}. In this case hierarchical  graphs are given by 
	$G=(\NN,\EN, \node: \NN\to \PPa(\NN), \st:\EN \to \funH(\NN))$.	 $\funH$  determines edge type. Typical choices for $\funH$ are $\Pot$ or $(\_)^*$ for hyperedges, $\PotDir$ for undirected edges or  directed edges with $\funX^2: \cSets \to \cSets \times \cSets$ with $\funX^2(N) = N\times N$. 
	For an example see	 Sect.~\ref{ss.hierG}.\ref{p.bkk}.\\
	We use a coalgebra over $\rGF_1: \cSets \times \cSets \to \cSets \times \cSets$ with $\rGF(\NN,\EN) = \PPa (\NN) \times  \funH(\NN)$, then $(\coalg{\rGF_1},\M)$ is an $\M$-adhesive category.	
	\item A hierarchy where the edges are refined by subnets is obtained by the neighbouring function
	     $\st:\EN \to (\NN)^* \times \Potom (\NN)$ that maps edges to a pair  where the first component defines the incident nodes and the second component defines the nodes contained  by the edges. This nesting is layered as it is defined by the functor $\Potom$, see Def.~\ref{d.potom_set}. The resulting graphs are given by
			$G=(\NN,\EN, \st:\EN \to \NN^* \times \Potom(\NN))$.			
			The category of such graphs is given by the comma category $< Id_\cSets \downarrow \funG>$ with the functor  $\funG=   ((\_)^* \times \Potom) \circ \funX^2$. \\     
Note $\funG(\NN) =    ((\_)^* \times \Potom) \circ \funX^2 (\NN) = ((\_)^* \times \Potom) (\NN,\NN) = (\NN)^* \times \Potom (\NN)$.
\item For hierarchies, where the edges between nodes may have other parents than the nodes and where the edges may contain  subgraphs  (as in \cite{Palacz04}) the graphs can be given by the functions $\node: \NN \to \PPa(\NN \uplus \EN)$ and $\st:\EN \to \Pot(\NN) \times \PPa(\NN\uplus \EN)$. We use then a coalgebra with $\rGF_2(\NN,\EN) = (\PPa (\NN \uplus \EN), \Pot(\NN) \times \PPa(\NN \uplus \EN)) $. $\PPa (\NN \uplus \EN)$ yields nested sets of nodes and edges and $\Pot(N)$ yields the incident nodes of an hyperedge. To obtain an $\M$-adhesive category we construct $\rGF$ from other functors that yield the $\M$-adhesive category $\coalg{\rGF_2}$. 
\item Multiple hierarchies can be constructed as $\M$-adhesive categories using a copying functor $\funX^i: \cSets \to \cSets\times \cSets\times ... \times \cSets$. The the containment function $\node:\NN \to
\prod   \circ \funX^i \circ \PPa(\NN)$ yields for each node $i$ different nestings. For edges we may use hyperedges $\st:\EN \to \Pot(\NN)$. The corresponding $\M$-adhesive category $\coalg{\rGF_3}$ is given by 
$\rGF (\NN,\EN) = (\prod   \circ \funX^i \circ \PPa(\NN), \Pot(\NN))$ and corresponds to the multi-hierarchical  graphs in Sect.~\ref{ss.mulHier}. 
\item For bigraphs, see Sect.~\ref{ss.bigr} we use the following functions $\node:\NN \to \PPa(\NN)$
and $\st: \EN \to \Pot(\NN \uplus \EN)\times \Pot(\NN \uplus \EN)$. Again we obtain an $\M$-adhesive category $\coalg{\rGF_4}$ with   $\rGF_4(\NN,\EN) = (\PPa(\NN), \Pot(\NN \uplus \EN)\times \Pot(\NN \uplus \EN) )$ constructed from other functors.
\item \label{i.group} The functions $\node: \NN \to \PPa(\NN)$ and   $\st:\EN \to \NN \times \NN \times \PPa(\EN)$ 
allow  the description of graph grouping and give rise to the category of  coalgebraic graphs $\coalg{\rGF_5}$ with
$\rGF_5 (\NN,\EN) = (\PPa(\NN), \NN \times \NN \times \PPa(\EN))$ that corresponds roughly to the the graph grouping in Sect.~\ref{ss.graphgroup}.
\end{enumerate}

Summarizing, we have:\nopagebreak
\renewcommand{\arraystretch}{2}
\begin{table}[H]
\begin{tabular}{|c|m{47mm}|m{41mm}|m{38mm}|}
  \hline 
&Definition          
    &Categorical construction
		&Description \\ \hline \hline 
   1&\parbox{ 47mm}{$\st:E \to \PPa(N)$ }
   & \parbox{41mm}{comma category \\$< Id_\cSets \downarrow \PPa>$}
   & \parbox{38mm}{~\\[1mm]hyperedges between nodes and containers of nodes, but container have no explicit name, 
	                                see Ex.~\ref{ex.PPa_graph}\\[-3mm]~}  
			\\  \hline 
2&\parbox{ 47mm}{$\node: \NN \to \PPa(\NN)$\\$\st: E \to \funH(\NN)$}
   & \parbox{41mm}{~\\[1mm]coalgebra $\coalg{\rGF_1}$\\$\rGF_1= \PPa \times \funH$} 
	 & \parbox{38mm}{~\\hierarchical graphs, \\
	                     $\funH$  determines edge type,\\
												  see Sect.~\ref{ss.hierG}.\ref{p.bkk}\\[-3mm]~} \\  %\cline{2-3}
  &\multicolumn{2}{|c|}{ for $\funH \in\{\Pot,  \PotDir, \_\times\_,  (\_)^* \}$}
  & \parbox{38mm}{}  \\ \hline 	
3&\parbox{ 47mm}{$\st: E\to \NN^* \times \Potom(\NN)$}
 & \parbox{41mm}{~\\[1mm]comma category \\$< Id_\cSets \downarrow \funG>$ with\\ 
                                        $\funG=  ( (\_)^* \times \Potom) \circ \funX^2$\\[-1mm]~
 } 
&  \parbox{38mm}{~\\hierarchical graphs,\\
	                    see Sect.~\ref{ss.hierG}.\ref{p.DHP}\\[-3mm]~}\\  \hline 
4&\parbox{ 47mm}{$\node: \NN \to \PPa(\NN \uplus \EN)$\\$\st:\EN \to \Pot(\NN) \times \PPa(\NN\uplus \EN)$}
   & \parbox{41mm}{~\\[1mm]$\coalg{\rGF_2}$ with
	                                                                         $\rGF_2 =$\\
	                                                                         $(\PPa\circ \coprod)   \,\times \,  $\\\hspace*{\fill}
	                                                                         $(\Pot  \times (\PPa\circ \coprod) )\circ (\funX^2\times Id_\cSets)$\\[-1mm]~} 
	 & \parbox{38mm}{~\\hierarchical graphs 
	                          see Sect.~\ref{ss.hierG}.\ref{p.pal}\\[-3mm]~} \\ \hline
5&\parbox{ 47mm}{$(\node_i)_{i<n}: \NN \to \PPa(\NN \uplus \EN)$\\$\st:\EN \to \Pot(\NN) $}
  & \parbox{41mm}{~\\[1mm]$\coalg{\rGF_3}$ with\\
$\rGF_3 =  (\prod   \circ \funX^i \circ \PPa\circ \coprod) \times \Pot$ \\[-1mm]~}
	& \parbox{38mm}{~\\multi-hierarchical  graphs,\\ 
	                          see Sect.~\ref{ss.mulHier}\\[-3mm]~} \\ 
 \hline	
6&\parbox{ 47mm}{~\\$\node:\NN \to \PPa(\NN)$\\$\st: \EN \to \Pot(\NN \uplus \EN)\times \Pot(\NN \uplus \EN)$\\[-3mm]~}
  & \parbox{41mm}{~\\[1mm]$\coalg{\rGF_4}$ with   \\
	$\rGF_4 =$\\\hspace*{\fill}
	                                                                         $ (\PPa(\NN), \Pot(\NN \uplus \EN)\times \Pot(\NN \uplus \EN) )$  \\[-1mm]~}
  & \parbox{38mm}{bigraphs, see Sect.~\ref{ss.bigr}}\\ 
	\hline 
7&\parbox{ 47mm}{~\\$\node: \NN \to \PPa(\NN)$ \\
                                $\st:\EN \to \NN \times \NN \times \PPa(\EN)$\\[-3mm]~}
  & \parbox{41mm}{~\\[1mm]$\coalg{\rGF_5}$with\\$\rGF_5 (\NN,\EN) = (\PPa(\NN), \NN \times \NN \times \PPa(\EN))$\\[-1mm]~} 
	& \parbox{38mm}{graph grouping, \\see Sect.~\ref{ss.graphgroup}}\\  \hline 
\end{tabular}
	\caption{Involved functors}
	\label{tab:InvolvedFunctors}
\end{table}

\subsection{Properties of Nested Nodes and Edges}

\begin{definition}[Properties of nested nodes]\label{d.proprecnodes}
\begin{enumerate}
	\item  \label{d.recnodes.uni} Nodes are unique if $c$ is injective.
	\item \label{d.recnodes.atomic} Vertices are the atomic nodes that  refer to themselves: $\aV=\{n \mid c(n) =n\}$ 
	\item \label{d.recnodes.atomic} Nodes are containers if $c(n) \in \PPa(\NN) -\NN$ 
	\item 
 The set of nodes is well-founded if and only if
 \begin{itemize}
	 \item \label{d.recnodes.base1} $X \in \NN \land Y\in \node(X)$ implies, that $Y\in \node(\NN)$
  \item  \label{d.recnodes.base2} $X \in \node(\NN) \land Y\in (X- \NN)$ implies, that $Y\in \node(\NN)$
\end{itemize}
\item 
 The set of nodes is hierarchical if and only if  $\node(n) \cap \node(n') \neq \emptyset$ implies $n=n'$.
\end{enumerate}
\end{definition}

\begin{definition}[Properties of nested edges]\label{d.proprecedges}
\begin{enumerate}
  
\item The set of atomic  hyperedges $\aE:= \{e \in \EN \mid \st(e)  \in \Pot(\NN)\}$.
\item Edges  are   node-based if the function $\stp{}: \EN \to \Pot(\NN)$ defined by  
$\stp{}(e) =   \{n \in \NN  \mid  n \in \st(e) \} \cup \bigcup_{x \in \st(e)} \stp{}(x)$ is well-defined.
\item Edges are atomic if they   are
node-based and if the function $\stp{}(E) \subseteq \aV$ only yields  vertices.
\end{enumerate}
\end{definition}

Analogously the properties for (un-)directed edges.

\section{Transformations of Hierarchical Graphs}
\label{s.hierGra}
% hierGra.tex%relWork.tex  vormals state

Here  we argue to what extent known concepts can be considered as $\M$-adhesive categories of hierarchical graphs. The detailed, mathematical investigation of each of these examples is beyond the scope of this paper.\\
Labels and attributes are not  considered in this paper, but labelled  or attributed graphs yield  $\M$-adhesive categories  (see \cite{FAGT,EGH10}) and at least labels can be introduced into  coalgebraic constructions (see \cite{Ru00,AD05}). 

\subsection{Hierarchical Graphs} \label{ss.hierG}
Many possibilities to define hierarchical graphs have already been investigated, e.g. \cite{BH01,DHP02,Busatto02,Palacz04,BKK05,BCM10}. 
	In \cite{ES95} the possibility of infinitely recursive hierarchies has already been introduced as  an infinite number of type layers.
Here we sketch how three of them, namely \cite{DHP02}, \cite{BKK05}  and \cite{Palacz04}, can be considered in this framework. 
\begin{enumerate}
\item \textit{Hierarchical Graphs as in \cite{BKK05}} \label{p.bkk}\\
	 In this approach graphs are grouped into packages via a coupling graph.
	A hierarchical graph is a system $H =(G,D,B)$, where $G$ is a graph some graph type, $P$ is a rooted directed acyclic graph, and $B$ is a bipartite 
coupling graph whose partition contains the nodes of $N_G$ and of $N_P$. All edges are oriented from the first $N_G$ to the second set
of nodes $N_P$ and  every
node in $N_G$ is connected to at least one node in $N_P$. For this approach we can consider coalgebraic graphs in  the coalgebra category $\coalg{\rGF_1}$ (see Table~\ref{tab:InvolvedFunctors}) with 
	$\node: \NN \to \PPa(\NN)$ being   well-founded. Additionally  a
 completeness condition,  stating that each atomic node is within some package, has to hold: \\
	$\forall n\in \NN: \node(n) = n \impl \exists p\in \NN: n \in \node(p)$ 
\end{enumerate}

\begin{wrapfigure}[10]{l}{.48\linewidth}~\\[-11mm]
	 \includegraphics[width=\linewidth]{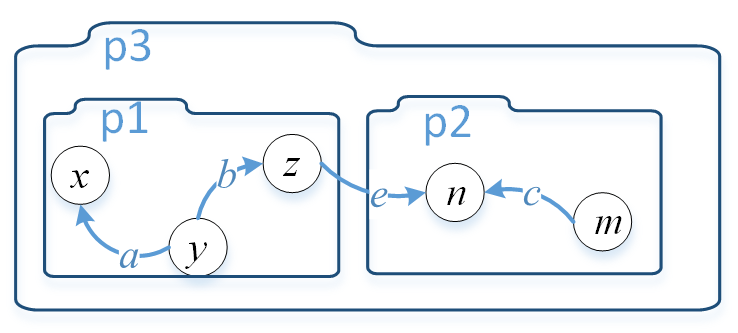}
	\caption{Hierarchical graph as  in \cite{BKK05}}
		\label{fig:ExBKK}
    \end{wrapfigure} 
	\noindent
	The 	packages are the nodes that are not atomic. The edge function is given by 
	$\st: E \to \funH(N)$ where $\funH(N)$ determines  the type of the underlying graphs.
	In Fig.~\ref{fig:ExBKK}  we have an example with two packages, that uses directed egdes.
	So based on $\funH=\funX^2$ we can give this example as a coalgebraic graph.   \\
	We have $\NN=\{n,m,x,y,z,p1,p2,p3\}$
	 \\
	\hspace*{5mm} with $\node(v) = \begin{cases}
	                    v&\text{; if } v\in\{n,m,x,y,z\}\\
                \{x,y,z\}&\text{; if } v =p1  \\
                \{n,m\} &\text{; if } v = p2 \\
                \{p1,p2\}&\text{; if } v =p3
	 \end{cases}
	$ and 
	$\st: \begin{cases}
	                    a\mapsto (y,x)\\
	                    b\mapsto (y,z)\\
	                    c\mapsto (m,n)\\
	                    e\mapsto (z,n)
	 \end{cases}
	$

\begin{enumerate}
	\setcounter{enumi}{1}
	\item \textit{Hierarchical Hypergraphs as in \cite{DHP02}} \label{p.DHP}
Hypergraphs $H=(V, E, att, lab)$ in \cite{DHP02} 
consist of two finite sets $V$ and $E$ of vertices and hyperedges. These are equipped with an order,  so the attachment function is defined by $att: E \to V^*$. 
The hierarchy  is given in layers, in the sense that subsets  in the same layer have the same nesting depth. So, edges are within one layer.
	Hierarchical graphs $<G,F,cts: F \to \mathcal{H} > \in \mathcal{H}$ are given with special edges $F$ that contain potentially hierarchical subgraphs.
Fig.~\ref{fig:ExDHP} depicts a hierarchical graph that can be considered to be a graph  in the comma category $< Id_\cSets \downarrow \funG>$  (see Table~\ref{tab:InvolvedFunctors}). The graph 
  $G=(\NN,\EN,\st)$ with $\st: E\to \NN^* \times \Potom(\NN))$ is defined so that edges are node-based. 
	\\[-8mm]
	\begin{figure}[H]
	\hspace*{5mm}
	 \begin{subfigure}[b]{0.5\linewidth}
	\includegraphics[width=\linewidth]{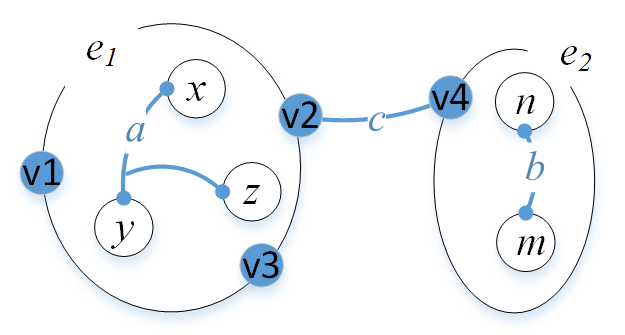}
        \caption{ as given in\cite{DHP02}}
				\label{fig:ExDHP}
    \end{subfigure} \hfill
	 \begin{subfigure}[b]{0.4\linewidth}
	$\begin{array}{rr@{\mapsto}l}
 \st: & a &  <xyz,\emptyset>\\
              & b & <nm,\emptyset>\\
              & c & <v2v4,\emptyset>\\
              & e_1 & <v1v2v3,\{x,y,z\}>\\
              & e_2 & <v4,\{n,m\}> 
	 \end{array}
	$
	\caption{as a graph in $< Id_\cSets \downarrow \funG>$}
    \end{subfigure}
		\caption{Example of hierarchical hypergraphs}
	\end{figure}

	\item \textit{Hierarchical Graphs as in \cite{Palacz04}}\label{p.pal}  are obtained from hypergraphs by adding a
	 parent assigning function to them. Nodes and edges can be assigned as a child of any other node or edge.These correspond to coalgebraic graphs in the category $\coalg{\rGF_2}$ (see Table~\ref{tab:InvolvedFunctors}). The parent function coincides with $\node: \NN \to \PPa(\NN \uplus \EN)$ being well-founded and hierarchical and $\st:\EN \to \Pot(\NN) \times \PPa(\NN\uplus \EN)$ since edges can have children as well.\\[-6mm]
	\begin{figure}[H]
	\hspace*{5mm}
	 \begin{subfigure}[b]{0.35\linewidth}
	     \includegraphics[width=\linewidth]{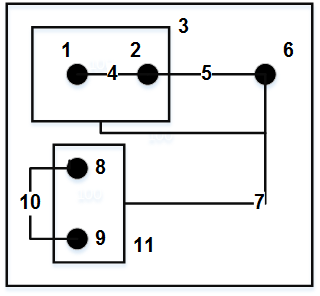}
		    \caption{Example from  \cite{Palacz04}}
    \end{subfigure} \hfill
	 \begin{subfigure}[b]{0.5\linewidth}
			\includegraphics[width=.9\linewidth]{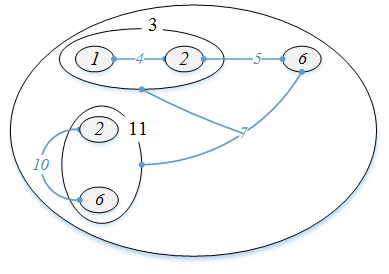}
		\caption{corresponding coalgebraic graph}
    \end{subfigure} \hfill
		    \caption{Hierarchical graph  in \cite{Palacz04}}
				\label{fig:ex_hierG_Pal04}		
	\end{figure}
In Fig.~\ref{fig:ex_hierG_Pal04} the  nodes $\NN=\{1,2,3, 6,8,9,11\}$ and  the contains function $\node:\NN \to \PPa(\NN \uplus \EN)$, 
	yield the nodes and their children. The  hyperedges $\EN=\{4,5,7,10\}$ with $\st: \EN \to \Pot(\NN \times \PPa(\NN \uplus \EN))$
	 yield the edges. Note in this example the edges  are not nested. \\
	Contains and neighbour function are given by \\
$\begin{array}[t]{rr@{\mapsto}l@{\;\; } r@{\mapsto}l}
	    \node:
			  &  1&1  &8 &8 \\
			  & 2&2   &9 &9\\
			  & 3 &\{ 1,2,4\}& 11& \{8,9\}\\
			  &6&6    
	 \end{array}
	$\hfill and 
	 $\begin{array}[t]{rr@{\mapsto}l}
	    \st:&  4 &(\{1,2\}, \emptyset )\\
			      & 5 & (\{2,6\}, \emptyset )\\
			      & 7 & (\{3,6,11\}, \emptyset )\\
			     &10&(\{8,9\}, \emptyset )
	 \end{array}
	$
\end{enumerate}

	\subsection{Multi-Hierarchical Graphs} \label{ss.mulHier} 
	In \cite{SLPPG2017} multiple hierarchies have been suggested, first ideas can be  found in  \cite{Palacz04}. A finite set of child nesting functions is specified that relate nodes to set of nodes and edges. This corresponds to a finite family $(\node_i: \NN \to \PPa(\NN \uplus \EN))_{i<n}$ that are well-founded and hierarchical. For transformations of multi-hierarchical graphs  there is  the $\M$-adhesive category $\coalg{\rGF_3}$ of coalgebraic graphs  (see Table~\ref{tab:InvolvedFunctors}).

	\subsection{Bigraphs as an hierarchy} \label{ss.bigr}
Bigraphs \cite{Milner06} originate in process calculi for concurrent systems and provide a graphical model of computation.  A bigraph is composed of two
graphs: a place graph and a link graph. 
They emphasize interplay between physical locality and virtual connectivity. 
Reaction rules allow the reconfiguration of bigraphs.  A bigraphical reactive system consists of a set of bigraphs and a set of reaction rules, which can be used to reconfigure the set of bigraphs. Bigraphs may be  composed  and have a bisimulation that is a congruence wrt. composition.   Categorically, bigraphs are given as morphisms in  a symmetric partial monoidal category where the objects are  interfaces.  This construction corresponds to ranked graphs as given in  \cite{GH97} where morphisms are given by a isomorphism class of concrete directed graphs with interfaces. \cite{Ehr02_bigraphs} discusses extensively the relation of bigraphs to graph transformations. In 
\cite{GRJD16} a functor that flattens bigraphs into ranked graphs is provided that encodes the topological structure of the place graph into the node names. 
In \cite{BMPT14} bigraphs are shown to be essentially the same as gs-graphs that  present the place and the link graph  within one graph. 
We also represent bigraphs within one graph, where the hierarchical structure is given by a superpower set of nodes and the link structure is given by nested hyperedges.
Here we abstract  from the categorical foundations and give bigraphs as a special cases of hierarchical graphs.
Hence, we ignore their categorical structure, but we obtain a transformation system.
Nevertheless, often only the graphical representation of bigraphs is used \cite{WW12,Worboys13,BCRS16}.
\\
A bigraph is a 5-tuple:
$(V,E,ctrl,prnt,link):\langle k,X\rangle \to \langle m,Y\rangle $,
where ${\displaystyle V}$ is a set of nodes, $E$ is a set of edges, $ctrl$ is the control map that assigns controls to nodes, $prnt$ is the parent map that defines the nesting of nodes, and $link$ is the link map that defines the link structure.
The notation  $\langle k,X\rangle \to \langle m,Y\rangle$  indicates that the bigraph has $k$ holes (sites) and a set of inner names $X$ and $m$ regions, with a set of outer names $Y$. These are respectively known as the inner and outer interfaces of the bigraph.\\\
Below we illustrate the relation of bigraphs to coalgebraic graphs in $\coalg{\rGF_4}$ (see Table~\ref{tab:InvolvedFunctors}) in an example.
\begin{figure}[h]
    \begin{subfigure}[b]{.5\textwidth}
		    \includegraphics[width=\linewidth]{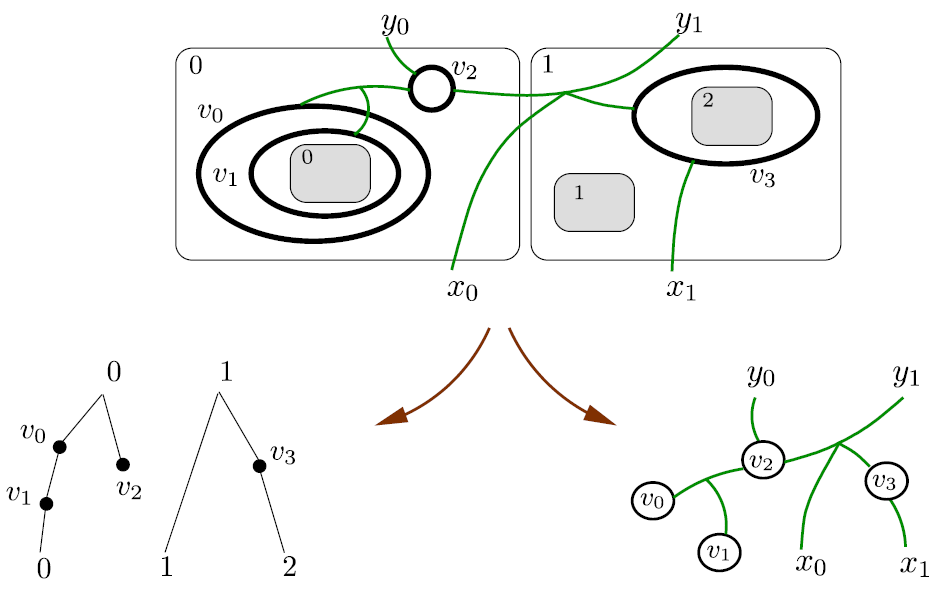}
\caption{Example from \cite{Milner06}}
  \label{fig:pureBiGr_fig3}
\end{subfigure}
\hfill
    \begin{subfigure}[b]{.5\textwidth}
		\includegraphics[width=\linewidth]{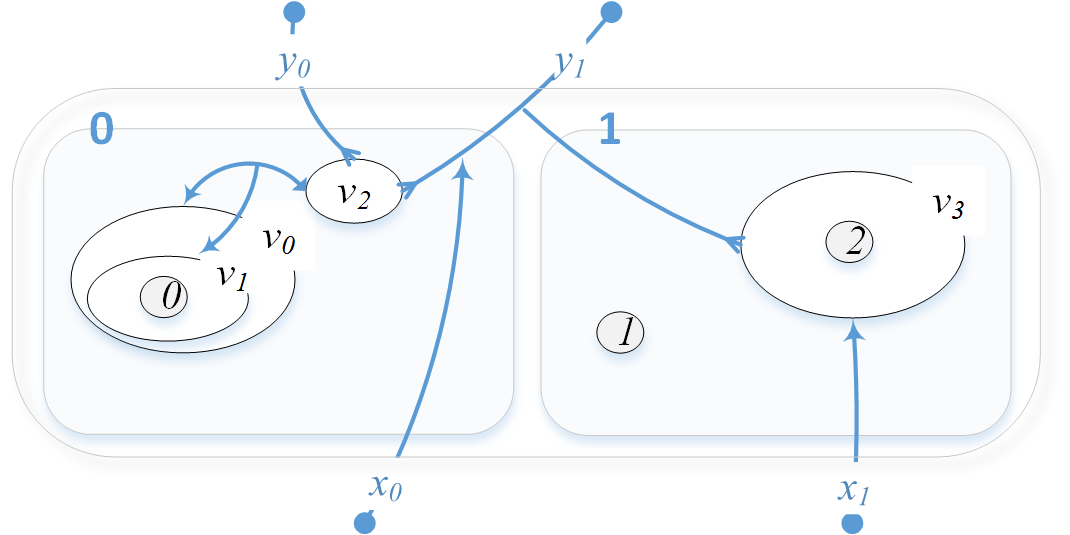}
\caption{as a graph in $\coalg{\rGF_4}$}
\label{fig:bigraph_ex}
\end{subfigure}
\caption{Bigraph}
\end{figure}
In Fig.~\ref{fig:pureBiGr_fig3} we have an introductory example from \cite{Milner06} that we represent as a  coalgebraic graph. 
The  coalgebraic graphs in the $\M$-adhesive category $\coalg{\rGF_4}$  need to have  well-founded and hierarchical nodes, where the contains function represents the parent function, so $\node = prnt$.
The function $cntrl$ yields basically the in- and out-degree of each node.  The $link$ function yields hyperedges, which we represent as directed hyperedges.
Hyperedges connecting outer names are represented as  directed hyperedges  with the arc itself as the target, those
 connecting inner names  as  directed hyperedges  with the arc itself as the source. The regions correspond to the roots $\mathbf{0}, \mathbf{1}$ of the forests
given by $\node$ and the site are the distinguished atomic nodes $0,1,2$.
The  nodes $\NN=\{ \mathbf{0}, \mathbf{1}, v_0,v_1,v_2,v_3, 0,1,2\}$ and  the contains function $\node:\NN \to \PPa(\NN)$, 
	yield the place graph. The directed nested hyperedges $\EN=\{e_1,e_2,e_3,e_4,e_5\}$ with $\st: \EN \to \Pot(\NN \uplus \EN)\times \Pot(\NN \uplus \EN)$ 
	 yield the link graph. We have:\\ 
$\begin{array}[t]{rr@{\mapsto}ll}
	    \node:
			  & \mathbf{0} &\{v_0,v_2\}\\
			  & \mathbf{1} &\{v_3,1\}\\
			  & v_0 &\{ v_1\}\\
			  & v_1 &\{ 0\}\\
			  & v_2 & v_2\\
			  & v_3 &\{2\}\\
			  & i &i 		\text{;for } 0\le i\le2
	 \end{array}
	$
	 $\begin{array}[t]{rr@{\mapsto}l}
	   \text{ and }  \st:& e_1 & (\{v_1,v_2,v_3\}, \{v_1,v_2,v_3\}) \\
			     & y_0 &  (\{ v_2\}, \{ v_2\})\\
			     & y_1 &  (\{ v_2,v_3\},  \{ v_2,v_3\})\\
			     & x_0 &  (\{x_0\},  \{y_1\})\\
			     & x_1 &  (\{x_1\},  \{v_3\})	 
	 \end{array}
	$
	\\	
Assuming $\node$ to be just well-founded we obtain bigraphs with  sharing as in \cite{SC15}.

\subsection{Graph Grouping}  \label{ss.graphgroup} \cite{JPR17} aims at a fundamentally different application area, namely graph grouping to  support data analysts  making decisions based on very  large graphs. Here, a graph hierarchy is established to cope with large amounts of data and to aggregate them.   Graph grouping operators produce a so-called summary graph containing super vertices and super edges. A super vertex stores the properties representing the group of nodes, and a super edge stores the properties representing the group of edges. Basically this leads to a  contains function
$\node: \NN \to \PPa(\NN)$ 
	that are well-founded but not necessarily hierarchical and a neighbour function $\st:\EN \to \NN \times \NN \times \PPa(\EN)$.
	These can be given as coalgebraic graphs in the category of  coalgebras $\coalg{\rGF_5}$   (see Table~\ref{tab:InvolvedFunctors}) that is $\M$-adhesive.\\
	But clearly this graph grouping is only sensible for attributed graphs since these used to abstract the data.

\section{Related Work}
\label{s.relWork}
%relWork.tex  vormals state

\subsection{Recursive Sets}
A set $M$ of integers is said to be recursive (see e.g. \cite{Davis82}) if there is a total recursive function $f(x) $ such that $f(x)=1$ for $x \in M$ and f$(x)=0$ for $x \notin M$. Any recursive set is also recursively enumerable.
\\
Finite sets, sets with finite complements, the odd numbers, and the prime numbers are all examples of recursive sets. The union and intersection of two recursive sets are themselves recursive, as is the complement of a recursive set.

\subsection{Recursive Graphs}
\cite{Bean76} is concerned  with recursive function theory that are analogous to certain problems in chromatic graph theory and introduces the following definition of recursive graphs according to \cite{Rem86}.
A recursive graph $\cal{G} =(V,E)$ is recursive, if $V$, the set of vertices is a recursive subset of the natural numbers $\Nat$ 
and $E$, the set of edges is a recursive subset of $\Nat^{(2)}$, the set of unordered pairs  from $\Nat$.
These recursive graphs have an infinite amount of nodes that need to be computed by a recursive function. 

\section{Conclusion}
\label{s.conc}
%conc.tex

We have presented a novel approach to hierarchies in graphs and graph transformations. This approach supports the use of the mature and extensive  theory of algebraic graph transformations for graphs with many different and also uncommon hierarchy concepts.
The aim of our approach is not a generalisation of hierarchy concepts in graph transformation but a possibility to access  algebraic graph transformation for graphs with a wide  spectrum of hierarchy concepts. 

We have presented an approach to graphs that allows arbitrarily nested nodes and edges being attached  to nodes, sets of nodes and edges.
This gives rise to  an abstract notion of graphs based on different functors. 

The vision  is a clear and simple access that provides a potential user with the hierarchical technique that is most adequate for the purpose. This requires a much deeper treatment of the 
hierarchical concepts at the abstract categorical level as well as an intuitive representation of these concepts.

\bibliographystyle{alphadin}
\bibliography{recgra}

\begin{appendix}
\section{Appendix}
%stuff
\subsection{Functors Preserving Pullbacks (along Monomorphims)}
	\label{s.misc}
	\label{s.powerset}

	The functors $\prod$ and $\funFinal$ are essentially limit constructions and hence compatible with  a pullback construction.\\

	The coproduct in \cSets is the disjoint union and can be considered a functor from the functor category $[\categ{I}, \cSets]$ to \cSets where \categ{I} a  small, discrete category.
	\begin{lemma}[$\coprod:\categ{I} \to  \cSets$ preserves pullbacks]% along injective functions]
	\label{l.coprod_PB}
	\end{lemma}
	
	\begin{proof}
	Given  pullbacks (PB) in \cSets for the index category \categ{I}  and diagram (1).\\

	$\xymatrix@=10mm{
	           A_i \ar[r]|{\pi_{B_i}}  \ar[d]|{\pi_{C_i}}    \ar@{}[dr]|{(PB)}
		  & B_i   \ar[d]|{f_i}   \\
			   C_i \ar[r]|{g_i} 
			&  D_i
		}
	$ \hfill
	$\xymatrix@=15mm{
	   X \ar[dr]|h \ar@/_6mm/[ddr]|{\hat{f}_i}  \ar@/^3mm/[drr]|{\hat{g}_i}    &&\\
	       &   \coprod( A_i) \ar[r]|{\coprod(\pi_{B_i})}  \ar[d]|{\coprod(\pi_{C_i})}    \ar@{}[dr]|{(1)}  \ar@{}[ur]|(.3){(2)}  \ar@{}[dl]|(.3){(3)} 
		  & \coprod(B_i)  \ar[d]|{\coprod(f_i)}   \\
			 &  \coprod(C_i)\ar[r]|{\coprod(g_i)} 
			&  \coprod(D_i)
		}
	$ \\[2mm]
	$(1)$ commutes since $\coprod$ is a functor.\\
	For each $X$ with  $g_i \circ \hat{f}_i  = f_i \circ \hat{g}_i$ there is the unique $h: X\to \coprod( A_i)$ with
	$h(x) = (c_i,b_i)$ with $\hat{g}_i (x) =b_i$ and $\hat{f}_i (x) = c_i$, so that $(2)$ and $(3)$ commute.
	\end{proof}

	The copy functor $\funX^n$ takes one set $S$ and copies the set $n$-times yielding an object $(S,S, ..., S) $ in the category $\cSets^n$.
	\begin{definition}[Copy functor $\funX$]~\\
	  $\funX^1 = Id_\cSets: \cSets \to \cSets$ and \\
		$\funX^{n+1} : \cSets \to \cSets^{n+1}$ with $\funX^{n+1} (S) = \funX^n(S) \times S$ for  sets $S$
		and \\
		$\funX^{n+1}(f) = \funX^n(f) \times f$ for functions $f:S\to S'$.\\
		
		 If $n$  is not  in the focus we may omit  it.
	
	\end{definition}
	Obviously, $\funX$ is well-defined and preserves injections. \\
	Note, $\funX$ differs  from a discrete schema $\mathcal{S}$ since $\mathcal{S}$ chooses sets that may be different. \\
	The symbol "$\times$" is used here to construct tuples of sets, whereas the product functor $\prod$ yields tuples of elements, so in these terms $\prod \circ\funX^2(S)=  S\times S$ with $\times$ the set-theoric cartesian product.

\begin{lemma}[The copy functor $\funX$ preserves pullbacks.]\label{l.funX_PB}
\end{lemma}
	
\begin{proof}

	Given  pullback $(PB)$  in $\cSets$:\\

	$\xymatrix@=10mm{
	           A\ar[r]|{\pi_{B}}  \ar[d]|{\pi_{C}}    \ar@{}[dr]|{(PB)}
		  & B  \ar[d]|{f}   \\
			   C \ar[r]|{g} 
			&  D
		}
		$
	 \\
	leading to the following pullback in the category  $\cSets^n$ :
	$$\xymatrix@R=10mm@C=25mm{
	   X \ar[dr]|h \ar@/_6mm/[ddr]|{\hat{f}}  \ar@/^3mm/[drr]|{\hat{g}}    &&\\
	       & (A,A,..,A) \ar[r]  \ar[d]   \ar@{}[dr]|{(1)}  \ar@{}[ur]|(.3){(2)}  \ar@{}[dl]|(.3){(3)} 
		  & (B,B,...,B)\ar[d]|{(f,f,...,f)}\\
			 & (C,C,...,C)\ar[r]|{(g,g,...,g)}
			&  (D,D,...,D)
		}
	$$
	$(1)$ commutes since $\funX$ is a functor.\\
	For each $X=(X_1, X_2,...,X_n)$ with  $(g,g,...,g) \circ \hat{f}  = (f,f,...,f) \circ \hat{g}$ there is the induced pullback morphism 
	$h: X\to  (A,A,...,A)$ with
	$h= (h_1,h_2,...,h_n) $ and  $h_i: X_i \to A$ are the induced pullback morphism of $(PB)$, so that $(2)$ and $(3)$ commute.
	\end{proof}

	Next we investigate some power set functors.
	
\begin{definition}[$\Pot,\Pot^{(i,j)}\Potfin$] \label{d.pots}
$\Pot,\Pot^{(i,j},\Potfin:\cSets \to \cSets$ with
\begin{itemize}
	\item $\Pot(M)= \{M' \subseteq M\}$,
	\item $\Pot^{(i,j)}(M)= \{M' \subseteq M \mid i\le |M|\le j\}$ and 
	\item$ \Potfin(M)= \{M' \subseteq M \mid  |M|\in \Nat\}$
\end{itemize}

so that $f:M\to N$ with $\Pot(f)(M') = f(M')$,  the same for $\Pot^{(i,j)}(f)$ and $ \Potfin(f)$
\end{definition}

%\subsection{Powerset Functor does not preserve pullbacks}
Contrary to Lemma A.39 in \cite{FAGT} the covariant powerset functor $\Pot$ does not preserve pullbacks.

Given a pullback diagram $(PB)$  in  \cSets with $D=\{d \}$, $B=\{1,2\}$ and $C=\{c,c'\}$.\\
Then $A=\{(1,c), (2,c), (1,c'), (2,c')\}$ is PB with the corresponding projections.

	 $(P,\pi_{\Pot(B)}, \pi_{\Pot(C)})$ is the pullback in \cSets over $(\Pot(D), \Pot(f_1), \Pot(g_1))$.\\
	$P=\{(B,C) \mid f_1(B) =g_1(C) \} = \{(\emptyset,\emptyset) \} \cup \{ \{1,\}, \{2\} ,\{1,2\}\}  \times \{ \{c\},\{c'\},\{c,c'\}\}$.\\
	
	Unfortunately $|P| = 10 \neq 16 = | \Pot(A)|$.

	\begin{lemma}[$\Pot:\cSets \to  \cSets$ preserves pullbacks along injective functions]
	\label{l.Pot_PB_along_injc}
	\end{lemma}
	
	\begin{proof}
	Given a pullback (PB) in \cSets and diagram (1) as above with $g_1: C \hookrightarrow A$ injective.\\
 Pullbacks and  the powerset functor preserve injections, so \\
$\pi_B: A\hookrightarrow B$, $\Pot(\pi_B): \Pot(A) \hookrightarrow \Pot(B)$
	and $\pi_{\Pot(b)}: P \hookrightarrow \Pot(B)$ are injective.\\
	$\xymatrix@=15mm{
	           A \ar@{^{(}->}[r]|{\pi_B}  \ar[d]|{\pi_C}    \ar@{}[dr]|{(PB)}
		  & B   \ar[d]|{f_1}   \\
			   C \ar@{^{(}->}[r]|{g_1} 
			&  D
		}
	$ \hfill
	$\xymatrix@=15mm{
	            P   \ar@{..>}@/^/[dr]|{\bar{h}}   \ar@/^/@{^{(}->}[drr]|{\pi_{\Pot(B)} }  \ar@/_/[ddr]|{\pi_{\Pot(C)}  }     %\ar @{} @/_/  [drr]|{[2)}  
									&& 							     \\
       &    \Pot(A) \ar@{^{(}->}[r]|{\Pot(\pi_B)}  \ar[d]|{\Pot(\pi_C)}  \ar@{}[dr]|{(1)} \ar@/^/[ul]|h   \ar@{}[ur]^<<<{(2)} \ar@{}[dl]|<<<{(3)} 
		  &   \Pot(B)   \ar[d]|{\Pot(f_1)}
			\\
			 &    \Pot(C) \ar@{^{(}->}[r]|{\Pot(g_1)}
			&   \Pot(D)
		}
	$
	\\
	Since $(PB)$ is a pullback diagram we have  $A=\{ (b,c) \mid f_1(b) = g_1(c)\}$.\\
	 $(1)$ commutes, since \Pot  is a functor.\\
	$h: \Pot(A) \to P$ is given by $h(A') = (\Pot(\pi_B)(A'),\Pot(\pi_C)(A') )$ the induced PB morphism so that, 
	$\Pot(\pi_B) \circ h = \pi_{\Pot(B)} $ and    $\Pot(\pi_C) \circ h = \pi_{\Pot(C)} $ .\\

	We define  $\bar{h}: P \to \Pot(A)$ with \\
	$\bar{h}(X,Y) = \{ (x,y) \mid x \in X, \; y \in Y, \;  f_1(x) = g_1(y)\}  =( X\times Y) \cap A $   \\and  have:
	
	\begin{enumerate}
		\item 		$\bar{h}$ is well-defined since $\bar{h}(X,Y)  \in \Pot(A)$.
		\item $(2)$ commutes along $\bar{h}$, i.e.  $\Pot(\pi_B) \circ \bar{h}  = \pi_{\Pot(B)}(X,Y)$\\
		  For $(X,Y) \in P$ we have 	
			\begin{align*}
			    \Pot(\pi_B) \circ \bar{h} (X,Y)\\
					                 & =  \Pot(\pi_B)  \{ (x,y) \mid x \in X, \; y \in Y, \;  f_1(x) = g_1(y)\} \\
													&=  \{ x \mid x \in X \land  \exists; y \in Y : f_1(x) = g_1(y)\}) \\
													&=X \\
													& = \pi_{\Pot(B)}(X,Y)
					\end{align*}
	\end{enumerate}
	We show $P\cong \Pot(A)$:
	\begin{itemize}
	   \item $h \circ  \bar{h} (X,Y) = id_P$:\\ 
		For each $(X,Y) \in P$ we have
		\begin{align*}
h \circ  \bar{h}(X,Y) &= h(  \{ (x,y) \mid x \in, \; y \in Y, \;  f_1(x) = g_1(y)\})\\
													&= h(  \{ (x,y) \mid x \in X, \; y \in Y \;  f_1(x) = g_1(y)\})\\
													&=  (  \{ x \mid x \in X \land  \exists; y \in Y : f_1(x) = g_1(y)\}),  \{ y \mid y \in Y \land  \exists  x \in X, : f_1(x) = g_1(y)\})\\
													&=(X,Y) 
					\end{align*}
		\item $\bar{h} \circ h = id_{\Pot(A)}$:\\
		Since $\Pot(\pi_B)$ is injective $\Pot(\pi_B) \circ \bar{h} \circ h = \pi_{\Pot(B)}  \circ h = \Pot(\pi_B)  = \Pot(\pi_B)  \circ  id_{\Pot(A)}$\\
		             we have  $\bar{h} \circ h = id_{\Pot(A)}$
	\end{itemize}
	\end{proof}

	\begin{corollary}[$\Pot^{(i,j)}, \Potfin:\cSets \to  \cSets$ preserve pullbacks along injective functions]
	\label{c.Pot_PB_along_injc}
	\end{corollary}

\begin{lemma}[$\rGF$ preserves pullbacks along monos]\label{l.rGF_PB}
 Let the coalgebraic graph functor $\rGF: \cSets\times\cSets \to \cSets \times \cSets$ be given 
with $\rGF(\NN,\EN) = (\funF(\NN), \funG(\EN))$.  $\rGF$ preserves pullbacks along monomorphisms  provided $\funF:  \cSets\times\cSets \to \cSets $ and $\funG:  \cSets\times\cSets \to \cSets $  preserve injections and pullbacks along injective morphisms.
\end{lemma}
	
\begin{proof}

	Given  pullback $(PB1)$  and $(PB2)$ in \cSets with the momomorphisms $f_{\NN}$ and $f_{\EN}$.\\

	$\xymatrix@=10mm{
	           A_{\NN} \ar[r]|{\pi_{B_{\NN}}}  \ar@{^{(}->}[d]|{\pi_{C_{\NN}}}    \ar@{}[dr]|{(PB1)}
		  & B_{\NN}   \ar@{^{(}->}[d]|{f_{\NN}}   \\
			   C_{\NN} \ar[r]|{g_{\NN}} 
			&  D_{\NN}
		}
	$ \hfill
	$\xymatrix@=10mm{
	           A_{\EN} \ar[r]|{\pi_{B_{\EN}}}  \ar@{^{(}->}[d]|{\pi_{C_{\EN}}}    \ar@{}[dr]|{(PB2)}
		  & B_{\EN}   \ar@{^{(}->}[d]|{f_{\EN}}   \\
			   C_{\EN} \ar[r]|{g_{\EN}} 
			&  D_{\EN}
		}
	$
	 \\
	leading to the following pullbacks, due to the assumption for $\funF$\\
	$$\xymatrix@R=10mm@C=25mm{
	           \funF(A_{\NN},A_{\EN}) \ar[r]|{\funF(\pi_{B_{\NN}},\pi_{B_{\EN}})}  
						                                       \ar@{^{(}->}[d]|{\funF(\pi_{C_{\NN}},(\pi_{C_{\EN}})}    \ar@{}[dr]|{(PB3)}
		  & \funF(B_{\NN},B_{\EN})   \ar@{^{(}->}[d]|{\funF(f_{\NN},f_{\EN})}   \\
			   \funF(C_{\NN},C_{\EN}) \ar[r]|{\funF(g_{\NN},g_{\EN})} 
			&  \funF(D_{\NN},D_{\EN})
		}
$$ 

and due to the assumption for $\funG$\\
	$$\xymatrix@R=10mm@C=25mm{
	           \funG(A_{\NN},A_{\EN}) \ar[r]|{\funG(\pi_{B_{\NN}},\pi_{B_{\EN}})}  
						                                       \ar@{^{(}->}[d]|{\funG(\pi_{C_{\NN}},(\pi_{C_{\EN}})}    \ar@{}[dr]|{(PB4)}
		  & \funG(B_{\NN},B_{\EN})   \ar@{^{(}->}[d]|{\funG(f_{\NN},f_{\EN})}   \\
			   \funG(C_{\NN},C_{\EN}) \ar[r]|{\funG(g_{\NN},g_{\EN})} 
			&  \funG(D_{\NN},D_{\EN})
		}
$$ 
	Then we have in the category  $\cSets\times\cSets$ :
	$$\xymatrix@R=10mm@C=25mm{
	   X \ar[dr]|h \ar@/_6mm/[ddr]|{\hat{f}}  \ar@/^3mm/[drr]|{\hat{g}}    &&\\
	       & \rGF(A_{\NN},A_{\EN}) \ar[r]  \ar@{^{(}->}[d]   \ar@{}[dr]|{(1)}  \ar@{}[ur]|(.3){(2)}  \ar@{}[dl]|(.3){(3)} 
		  & \rGF(B_{\NN},B_{\EN})\ar@{^{(}->}[d]|{\rGF(f_{\NN},f_{\EN})}\\
			 &  \rGF(C_{\NN},C_{\EN})\ar[r]|{\rGF(g_{\NN},g_{\EN})}
			&  \rGF(D_{\NN},D_{\EN})
		}
	$$
	$(1)$ commutes since $\rGF$ is a functor.\\
	For each $X$ with  $\rGF(g_{\NN},g_{\EN}) \circ \hat{f}  = \rGF(f_{\NN},f_{\EN}) \circ \hat{g}$ there is the induced pullback morphism 
	$h: X\to  \rGF(A_{\NN},A_{\EN})$ with
	$(h_1,h_2) $, so that $(2)$ and $(3)$ commute. $h_1: X\to \funF(A_{\NN},A_{\EN})$ is the induced pullback morphism of $(PB3)$ for the projection of the first component
	and $h_2:X \to \funG(A_{\NN},A_{\EN})$ is the induced pullback morphism of $(PB4)$ for the projection of the second component.
	\end{proof}

	\subsection{Noetherian Induction}
	\label{s.NoetherianInd}
	A set $S$ together with a binary relation $R$ over $S$  is well-founded if and only if every non-empty subset $S' \subseteq S$ has a minimal element; that is, 
\[	\forall S'\subseteq S\ (S'\neq \emptyset  \impl  \exists x\in S'\;\;\forall s\in S'\;\,(s,x)\notin R)  \]
Examples of relations that are not well-founded include
the negative integers $\Int^-$  with the usual order, since any unbounded subset has no least element.  The powerset of a set  together with the inclusion of sets is a  well-founded set if and only if  the set is finite. Note, that the superpower sets are well-founded even for infinite sets, since we employ a different order.\\[2mm]
Well-founded sets  allow   Noetherian induction:\\
Given  a property ${\displaystyle P}$ for an order relation ${\displaystyle R }$  of a  wellfounded set $S$, then we have:
\begin{itemize}
	\item ${\displaystyle P(x)} $ holds  for all minimal elements of $S$.
   \item $x\in S$  and $P(x)$ holds  for all  $ yRx$ then  $P(x)$ holds.
\end{itemize}
Then  $P(x)$ holds for all $x \in S$.

\begin{lemma}[$\PPa(S)$ is well-founded]\label{l.PPa_wellfounded}
Let $R\subseteq \PPa(S) \times \PPb(S)$ be given by $A \;R \;B$ iff $np(A) \le np(B)$ with
$np:\PPa(S) \to \Nat$
\begin{itemize}
	\item $s \in S \impl  np(s) =0$
	\item $np(\emptyset) =1$
	\item $(X_i)_{i\in I} \in S \impl  np(\{ X_i | i \in I\}) = max \{ np( X_i) | i\in I\} +1$ with  $I \subseteq \Nat$
	%\item $X \subset \PPb(S) \impl np(\{X\})=  np(X) +1
\end{itemize}
\end{lemma}
\begin{proof}Due to the inductive definition of $\PPa$.
\end{proof}

\begin{lemma}[$\PPb(S)$ is well-founded]\label{l.PPb_wellfounded}
Let $R\subseteq \PPb(S) \times \PPb(S)$ be given by $A \;R \;B$ iff $np(A) \le np(B)$ with
$np:\PPb(S) \to \Nat$
\begin{itemize}
	%\item $s \in S \impl  np(s) =0$
	\item $np(\emptyset) =1$
	\item $(X_i)_{i\in I} \in S \impl  np(\{ X_i | i \in I\}) = max \{ np( X_i) | i\in I\} +1$ with  $I \subseteq \Nat$
	%\item $X \subset \PPb(S) \impl np(\{X\})=  np(X) +1
\end{itemize}
\end{lemma}
\begin{proof}Due to the inductive definition of $\PPb$.
\end{proof}

\begin{lemma}[$\Potom(S)$ is wellfounded]\label{l.Potom_wellfounded}
Let $R\subseteq \Potom(S) \times \Potom(S)$ be given by $A \;R \;B$ iff $np(A) \le np(B)$ with
$np:\Potom(S) \to \Nat$
\begin{itemize}
	\item $s \in S \impl  np(s) =0$
	\item $np(\emptyset) =1$
	\item $(X_i)_{i\in I} \in S \impl  np(\{ X_i | i \in I\}) = max \{ np( X_i) | i\in I\} +1$ with  $I \subseteq \Nat$
	%\item $X \subset \PPb(S) \impl np(\{X\})=  np(X) +1
\end{itemize}
\end{lemma}
\begin{proof}Due to the inductive definition of $\Potom$.
\end{proof}
\end{appendix}

\end{document}